\documentclass{article}

\usepackage{arxiv}

\newtheorem{theorem}{\bf Theorem}[section]
\newtheorem{proposition}[theorem]{\bf Proposition}

\usepackage[utf8]{inputenc}
\usepackage[T1]{fontenc}
\usepackage{hyperref}
\usepackage{url}
\usepackage{booktabs}
\usepackage{amsfonts}
\usepackage{nicefrac}
\usepackage{microtype}
\usepackage{relsize}
\usepackage{verbatim}
\usepackage{graphicx}
\usepackage{dcolumn}
\usepackage{bm}
\usepackage{float}
\usepackage{graphics}
\usepackage{color}
\usepackage{xcolor}
\usepackage{soul}
\usepackage{subfigure}
\usepackage{subcaption}
\usepackage{tikz}

\usepackage{tabularx}
\usepackage{array}
\usepackage{booktabs}

\newcolumntype{L}[1]{>{\raggedright\arraybackslash\hsize=#1\hsize}X}

\usetikzlibrary{arrows.meta,positioning,fit,calc}
\usepackage{booktabs}
\usepackage{amsmath,amssymb,amsfonts}

\usepackage{algorithm}
\usepackage{algpseudocode}
\usepackage{enumitem}
\setlist[itemize]{leftmargin=*}
\setlist[enumerate]{leftmargin=*}
\definecolor{rev}{rgb}{0,0,0}
\definecolor{rev2}{rgb}{0,0,0}
\usepackage{array}
\newcolumntype{P}[1]{>{\centering\arraybackslash}p{#1}}
\usepackage{multirow}
\usepackage{cancel}
\usepackage[capitalise]{cleveref}

\usepackage{cite}
\usepackage{tabularx}
\usepackage{threeparttable}
\usepackage{pifont}

\newcolumntype{Y}{>{\centering\arraybackslash}X}

\newcommand{\R}{\mathbb{R}}

\newcommand{\bs}[1]{\boldsymbol{#1}}
\newcommand{\vc}{\bs{c}}
\newcommand{\vx}{\bs{x}}
\newcommand{\vu}{\bs{u}}
\newcommand{\vp}{\bs{\pi}}          
\newcommand{\vzeta}{\bs{\zeta}}
\newcommand{\Gop}{\mathcal{G}}
\newcommand{\relL}{\mathcal{E}_{L^2}}
\newcommand{\Bnet}{\mathcal{B}_{\theta}}
\newcommand{\Tnet}{\mathcal{T}_{\theta}}

\title{Spectral-Embedded Operator Learning for Three-Phase Interfacial Flow: A Ternary Cahn-Hilliard-Navier-Stokes Benchmark}

\author{
  Muhammad Abid \\
  Department of Mechanical and Aerospace Engineering,\\
  University of Tennessee, Knoxville\\
  Knoxville, TN 37996, USA.\\
  \texttt{mabid@vols.utk.edu}
  \And
  Arth Sojitra \\
  Department of Mechanical and Aerospace Engineering\\
  University of Tennessee, Knoxville\\
  Knoxville, TN 37996, USA\\
  \texttt{asojitra@vols.utk.edu}
  \And
  Omer San \\
  Department of Mechanical and Aerospace Engineering,\\
  University of Tennessee, Knoxville\\
  Knoxville, TN 37996, USA.\\
  \texttt{osan@utk.edu}
}

\begin{document}
\maketitle

\begin{abstract}
Operator-learning surrogates have been benchmarked largely on single-field,
single-interface problems, leaving unclear whether architectural choices
validated in those settings transfer to constrained, multiphase flows. We
introduce a three-phase interfacial-flow benchmark to examine whether the trunk
coordinate representation matters for a multi-channel, interface-dominated
target. The configuration consists of an air bubble rising through water,
piercing a water-oil interface, and entraining a water plume into the oil
within a bounded, wall-confined domain. Reference data are generated using a
structure-preserving ternary Cahn-Hilliard-Navier-Stokes solver that
algebraically preserves the simplex constraint. From $1{,}024$ Sobol-sampled
simulations spanning a nine-dimensional parameter space, we learn the mapping
from physical parameters to five-channel space-time fields. We compare three
parameter-matched DeepONet variants differing only in trunk representation:
raw coordinates (DeepONet), random Fourier features (FEDONet), and a fixed
tensor-product Chebyshev dictionary (SEDONet). SEDONet reduces the test
relative $L^2$ error by $16.8\%$ compared with FEDONet and by $24.0\%$
compared with DeepONet, while improving all five output channels. Spatial and
temporal error analyses localize the principal gains near the diffuse
interfaces and after bubble breakthrough. The results indicate that the
Chebyshev representation is particularly effective for the strongly
non-periodic wall-normal and temporal structure of this three-phase flow.
\end{abstract}


\section{Introduction}
\label{sec:intro}
Flows in which three immiscible fluids share a domain are common in industrial
practice: enhanced oil recovery moves water, oil and gas through the same pore
network, froth flotation separates minerals at gas-water-solid contacts, and
geological $\mathrm{CO}_2$ sequestration and nuclear severe-accident analysis
involve at least three phases in simultaneous contact. Relative to two-phase
flow the ternary problem is a qualitative change in difficulty, for three
reasons. There are three pairwise interfaces instead of one, each with its own
tension $\sigma_{ij}$, and they meet along triple junctions whose geometry is
fixed by a Neumann-triangle balance, so contact angles are determined by the
three tensions jointly rather than independently \cite{degennes1985}. When the
tensions violate the partial-spreading inequalities, one phase films
spontaneously between the other two with no external forcing, so the parameter
space contains regime boundaries that a naive sampling strategy will cross
without warning \cite{boyer2006}. And the phase variables are algebraically
constrained, summing to unity pointwise, a condition any
faithful discretization, and we will argue any faithful surrogate, should
respect rather than merely approximate \cite{kim2012}.

The canonical model problem for this regime, and the one we adopt, is a gas
bubble rising through a lower liquid and crossing a horizontal liquid-liquid
interface. Even in isolation this configuration exhibits a rich sequence of
regimes: the bubble deforms into a spherical cap, decelerates as it approaches
the interface, drapes a film of the lower liquid over its crown, breaks
through, and drags a column of the lower liquid into the upper one, which
subsequently thins, necks and pinches. These crossing regimes have been mapped
experimentally and numerically as a function of Bond and Archimedes numbers,
showing that the entrained column and its breakup are governed by an
inertia-capillarity competition sensitive to all three tensions and both
density contrasts \cite{bonhomme2012}. The problem therefore packages, in one
compact two-dimensional domain, essentially every feature that makes multiphase
simulation hard: thin films, high curvature, a moving triple junction and a
genuine topological transition.

Two families of numerical methods can resolve such a flow, and the choice
between them is dictated by the topological transition. Sharp-interface
approaches, volume of fluid \cite{hirt1981}, level sets \cite{sussman1994},
front tracking \cite{tryggvason2001}, represent the interface explicitly and
reconstruct it at each step, achieving excellent geometric fidelity in
two-phase settings but requiring dedicated junction models \cite{smith2002}
and ad hoc surgery precisely when the topology changes. Diffuse-interface, or
phase-field, methods take the opposite route: the interface becomes a thin
transition layer of width $\varepsilon$ over which an order parameter varies
smoothly, and topological change is handled automatically because nothing
needs to be reconstructed. The formulation descends from the free energy of a
nonuniform system \cite{cahn1958}, its coupling to fluid mechanics has been
reviewed \cite{anderson1998}, its thermodynamic consistency at variable
density was settled in \cite{lowengrub1998, abels2012}, and the schemes used
at large density ratios follow \cite{jacqmin1999, ding2007}. For three or more
components, a free energy algebraically consistent with the three pairwise
tensions keeps the phase variables on the Gibbs simplex \cite{boyer2006}; it
is this construction we adopt, and its algebraic structure is what allows our
reference solver to make the simplex constraint exact rather than approximate.

The price of that fidelity is severe, and structural rather than incidental. In
a diffuse-interface formulation the transition layer must be resolved by
several mesh cells, so the interface thickness is tied to the grid and the grid
cannot be coarsened; and the explicit treatment of the interfacial forces
imposes a time-step restriction that tightens faster than the mesh is refined, a capillary-wave limit $\Delta t \sim \sqrt{\rho h^3/\sigma}$ in the
capillary-dominated regime \cite{brackbill1992}, and a viscous limit
$\Delta t \sim h^2/\nu$ where the viscosity contrast is large. Careful
treatment of the variable-density pressure coupling \cite{dodd2014} makes each
step cheap, but it cannot reduce their number: a single two-dimensional run
costs tens of thousands of steps. A parametric study, sweeping tensions,
density ratios, bubble size, layer thickness, multiplies that by the size of
the design, and it is precisely such studies that engineering practice demands.
This is the cost structure that surrogate modelling exists to break.

Machine-learning surrogates for fluid mechanics have matured rapidly toward
that goal; the landscape has been surveyed broadly
\cite{brunton2020, vinuesa2022}, and its intersection with interface capturing
specifically \cite{gibou2019}. Early work targeted particular tasks: hybrid
solvers in which a learned component corrects a coarse-grid discretization
\cite{kochkov2021}, and physics-informed neural networks that embed the PDE
residual in the training objective \cite{raissi2019pinn, karniadakis2021}. A
limitation these share is that they learn a function, one solution
field, on one discretization, for one set of conditions, so a parametric
sweep requires retraining, which reintroduces the cost the surrogate was meant
to remove.

Operator learning removes that limitation by targeting the solution map itself,
as a mapping between infinite-dimensional function spaces. Its theoretical
foundation is the operator universal approximation theorem
\cite{chenchen1995}, realized in practice by the Deep Operator Network
\cite{lu2021} through a branch--trunk factorization in which one network
encodes the input and another the query coordinate; the neural-operator
framework \cite{kovachki2023} gives an alternative construction based on
learned integral kernels, of which the Fourier Neural Operator
\cite{li2021fno} is the best-known instance. The approach has already been
deployed at scale on multiphase problems in porous media
\cite{wen2022ufno, wen2023nested} and at a planetary scale in weather
forecasting \cite{pathak2022, lam2023graphcast}, and the field is
consolidating around shared benchmarks \cite{takamoto2022pdebench} and
multi-physics pretrained models \cite{herde2024poseidon}. What is
conspicuously absent from those benchmark suites is a problem with three
phases, a moving triple junction, and an algebraic constraint on the state,
which is to say, the problem this paper contributes.

Within the branch-trunk architecture, the component that encodes the query
coordinate, the trunk is almost always a shallow multilayer perceptron
acting on raw $(x,y,t)$, and that choice inherits a well-documented pathology.
Neural networks exhibit spectral bias, fitting low-frequency content first and
fastest and high-frequency content last or not at all \cite{rahaman2019}; the
effect has been characterized in the Fourier domain as a frequency principle
\cite{xu2020frequency} and connected to the eigenspectrum of the neural
tangent kernel \cite{jacot2018ntk}. For interfacial flow the bias is maximally
badly placed. The bulk of a phase-field solution is nearly constant and
carries almost no information, while everything physically interesting lives
in transition layers whose width is the smallest length scale in the problem, exactly the content the network learns last. The standard remedy is to
replace the raw coordinate with a fixed high-dimensional dictionary. Random
Fourier features \cite{tancik2020}, whose kernel interpretation goes back to
the random-features construction for large-scale kernel machines
\cite{rahimi2007} and whose effect on the tangent-kernel spectrum has since
been analyzed \cite{wang2021eigenvector}, were transplanted into the DeepONet
trunk as FEDONet \cite{sojitra2025}. A sinusoidal dictionary, however, encodes
a periodic prior, whereas classical numerical analysis settled long ago that
the natural basis on a bounded, non-periodic interval is polynomial
\cite{boyd2001, trefethen2013}, the motivation for the Chebyshev-embedded
trunk of SEDONet \cite{abid2026}.

Two structural facts about the present target shape what a surrogate for it can
and should do, and both are easy to overlook. The first is that the state is
confined to an affine subspace. A learned model can satisfy such a constraint
approximately, by penalizing violations in the loss, or exactly, by building it
into the output map; for climate emulation the second route has been shown to
be both cheaper and more reliable when the constraint is linear, since it costs
no capacity and cannot be traded away against the data term \cite{beucler2021}.
The second is that a branch-trunk prediction is a linear combination of a
fixed number of learned basis functions, so its accuracy is bounded below by
how well a subspace of that dimension approximates the solution manifold,
the Kolmogorov width, whose behaviour for parametric PDEs has been surveyed
\cite{cohen2015} and which is known to decay slowly for advection-dominated
problems with moving features \cite{peherstorfer2022}. Together these say that
the algebraic constraint should be imposed and not learned, and that no
coordinate dictionary should be expected to buy orders of magnitude on a
transport-dominated target. Both predictions are borne out below.

This paper addresses two gaps. The first is that the evidence for spectral
trunks, while encouraging, is narrow: existing evaluations use single-field
targets on domains that are one spatial dimension plus time or two spatial
dimensions with no time, carrying no algebraic constraint on the state, no
competition between fields of different character and no genuinely
interface-dominated physics. Whether a trunk embedding still matters when the
target is a five-channel space-time field on a constrained manifold, dominated
by three mutually interacting interfaces and a topological transition, is an
open question. The second gap is subtler: because operator surrogates are
almost always evaluated by pointwise error norms alone, it is rarely possible
to tell where an architectural change helps. A reduction in relative
$L^2$ error of any size is compatible with a uniform improvement
everywhere, with a large improvement in a small region, or with an improvement
in the bulk purchased at the cost of the interfaces, entirely different
outcomes for interfacial physics, of which only the second is worth having. A
scalar norm also cannot discriminate between competing explanations of an
improvement, so a result reported that way leaves the reader unable to judge
whether it will transfer to any other problem.

Our contributions are as follows.

\begin{enumerate}
\item \textbf{A ternary CHNS operator-learning benchmark:} An air bubble rising
through water, breaking through a water-oil interface and entraining a plume.
The target is a five-channel space-time field $(c_w,c_o,c_a,u,v)$ over a
nine-dimensional parameter space, produced by a $1{,}024$-member Sobol
\cite{sobol1967} ensemble. To our knowledge this is the first operator-learning
benchmark with three phases and a simplex-constrained state.

\item \textbf{A structure-preserving reference solver:} Choosing the mobility as
$M_i = M_0/\Sigma_i$ makes $\sum_i M_i\mu_i \equiv 0$, so $\sum_i c_i = 1$ is an
identity of the semi-discrete system rather than something the solver is nudged
toward (\Cref{prop:simplex}). Any simplex violation exhibited by a surrogate is
therefore provably its own. Exact DCT-II diagonalization of the Laplacian, the
Dodd-Ferrante split \cite{dodd2014} and batched execution in JAX \cite{jax2018}
give $1{,}024$ runs in $37$ minutes.

\item \textbf{Graded truncation of the Chebyshev dictionary:} Extending the
trunk embedding of \cite{abid2026} from $(x,t)$ to $(x,y,t)$ exposes a defect in
its lexicographic crop, which in three coordinates retains full degree in the
last index and almost none in the first. Total-degree ordering yields a
degree-balanced dictionary at identical feature budget.

\item \textbf{A parameter-exact controlled comparison:} The three models share
branch and trunk depths, widths, optimizer, schedule and data, differing only in
the trunk embedding. FEDONet and SEDONet have identical parameter counts, so the
gap between them cannot be attributed to capacity.

\item \textbf{Localization of the gain:} We decompose the error spatially,
temporally, by channel and across the parameter space. Essentially the entire
SEDONet advantage sits at the diffuse interfaces and emerges only after
breakthrough. The decomposition also falsifies the obvious mechanism, the
error is smallest at the walls, and isolates an alternative, the
representation of monotone, non-periodic trends in the wall-normal and temporal
coordinates (\Cref{sec:res:accuracy}).
\end{enumerate}

The remainder of the paper is organized as follows. \Cref{sec:related} works
through the literature in more detail, \Cref{sec:model} states the ternary CHNS
model and establishes its structural properties, and \Cref{sec:opnet}
formulates the operator-learning problem, the three trunk embeddings and the
simplex closure. \Cref{sec:results} reports convergence, accuracy and the
spatial, temporal and per-channel error decompositions. \Cref{sec:conclusion}
summarizes and states limitations, and \Cref{sec:future} sets out what we would
do next.
\section{Related Work}
\label{sec:related}

\subsection{Interface capturing and three-phase flow}
\label{sec:related:cfd}

The distinction drawn in \Cref{sec:intro} between sharp- and diffuse-interface
methods deserves elaboration, because it is what determines whether a
three-phase configuration is tractable at all. Sharp-interface approaches
represent the interface explicitly or implicitly and reconstruct it at each
step: volume of fluid \cite{hirt1981}, level sets \cite{sussman1994} and front
tracking \cite{tryggvason2001} are the standard representatives, and all three
achieve excellent mass or geometric fidelity in two-phase settings. Extending
them to three phases is where the difficulty concentrates. Triple junctions
require either explicit junction models or coupled level-set systems with
projection to prevent vacuum and overlap \cite{smith2002}, and topological
change, the very event of interest here, requires ad hoc surgery, so the
method must be told in advance about the events it is being used to discover.

Diffuse-interface methods trade that difficulty for a resolution requirement.
Because the order parameter varies smoothly over a layer of width
$\varepsilon$, merging and pinching are automatic consequences of the dynamics
rather than special cases in the code; the cost is that $\varepsilon$ must be
resolved. Within this family the multi-component case has been developed along
two lines. A free energy algebraically consistent with the three pairwise
tensions keeps the phase variables on the Gibbs simplex \cite{boyer2006}, and
the coupled Cahn-Hilliard/Navier-Stokes system built on it is developed in
\cite{boyer2010}; related multi-component formulations are given in
\cite{kim2005, kim2012}. The feature of the Boyer-Lapuerta construction that
matters most here is not its free energy but its Lagrange-multiplier term,
whose particular prefactor makes a specific choice of mobility turn the simplex
condition into an identity; we exploit exactly that in
\Cref{sec:model:simplex}.

The numerical ingredients are individually standard, but their combination is
what makes a large ensemble affordable. Energy-stable time integration for
Cahn-Hilliard systems rests on convex splitting \cite{eyre1998} and, more
recently, on the scalar auxiliary variable framework \cite{shen2018sav}. The
velocity-pressure coupling follows the projection method \cite{chorin1968},
with the variable-density modification \cite{dodd2014} that preserves a
constant-coefficient Poisson operator, the property that makes a
transform-based solve viable at every step and hence that removes iteration
and run-to-run cost variability from the ensemble. The capillary force is most
conveniently written in continuum-surface form \cite{brackbill1992}, whose
ternary generalization requires no estimate of interface normals or of
triple-junction geometry and therefore degrades gracefully exactly where a
curvature-based force would fail.

\subsection{Operator learning for parametric PDEs}
\label{sec:related:opnet}

The branch-trunk factorization \cite{lu2021} and the integral-kernel framework
\cite{kovachki2023} define the two dominant families, and the choice between
them is not settled. Approximation-theoretic guarantees for the first are
available \cite{lanthaler2022}, and a controlled comparison of the two on
matched data has been reported \cite{lu2022fair}, a study whose methodology,
holding everything fixed except the component under test, is the one we adopt
here at the level of a single layer rather than a whole architecture.

The subsequent literature has enriched the landscape along several axes:
alternative bases such as wavelets \cite{tripura2023wno}; attention and
transformer mechanisms \cite{hao2023gnot}; convolutional operators with
continuous-discrete equivalence \cite{raonic2023cno}; geometry-aware and
resolution-independent formulations \cite{li2023geofno, bahmani2025rino};
multifidelity closure constructions \cite{howard2023multifidelity}; and
physics-informed variants that embed residuals or variational principles
directly \cite{wang2021pideeponet, goswami2022vdeeponet}. Of particular
relevance here is a small body of work on interface-dominated problems, in
which architectures are modified to accommodate discontinuities across an
interface \cite{wu2024ionet, bi2025xideeponet}, and on the mitigation of
spectral bias in multiscale operator learning \cite{liu2024spectralbias}. These differ from ours in strategy rather than in aim: they change the
architecture to accommodate the interface, whereas we change only the
coordinate representation. The two are complementary, and our results establish
how much of the gap a coordinate dictionary alone can close.

Operator surrogates are past the proof-of-concept stage on problems with
genuine industrial stakes: reservoir-scale $\mathrm{CO}_2$ plume migration with
U-FNO \cite{wen2022ufno} and its nested successor \cite{wen2023nested}, and
operational-quality weather forecasting with FourCastNet \cite{pathak2022} and
GraphCast \cite{lam2023graphcast}. It is worth noting what they have in common:
all four are single-phase or effectively two-phase, all are posed on periodic
or quasi-periodic domains, and none carries a pointwise algebraic constraint on
the state. The same is true of the shared benchmark suites
\cite{takamoto2022pdebench, herde2024poseidon}, which is why the family of
problems studied here is not represented in them.

\subsection{Coordinate representations and spectral trunks}
\label{sec:related:embed}

The limitation motivating a coordinate embedding is well characterized. Neural
networks fit low frequencies first and fastest \cite{rahaman2019}; the
phenomenon has been described in the Fourier domain as a frequency principle
\cite{xu2020frequency} and connected to the eigenspectrum of the neural tangent
kernel \cite{jacot2018ntk}, which supplies the mechanism: convergence along an
eigendirection is governed by its eigenvalue, so directions with small
eigenvalues are learned last regardless of how important they are to the
target. Random Fourier features \cite{tancik2020}, whose kernel interpretation
goes back to the random-features construction for large-scale kernel machines
\cite{rahimi2007}, reshape that spectrum, and the reshaping has been analyzed
directly \cite{wang2021eigenvector}. Within operator learning, FEDONet
\cite{sojitra2025} transplants the construction into the DeepONet trunk and
reports gains across a range of PDE families, with the largest improvements on
oscillatory, broadband problems such as Kuramoto-Sivashinsky, a pattern
that is itself informative, since those are the problems whose structure best
matches a sinusoidal prior.

That prior is the limitation. A dictionary built from functions periodic on the
embedding scale is well matched to tori and to broadband oscillatory content,
and poorly matched to bounded domains with Dirichlet or Neumann boundaries,
which describes most of computational mechanics, and describes the
wall-confined, strictly non-periodic domain of the present problem exactly.
Classical numerical analysis settled the alternative long ago: on a bounded
interval the natural orthogonal basis is polynomial \cite{boyd2001}. Chebyshev
expansions converge geometrically for analytic functions \cite{trefethen2013},
their nodes cluster near the endpoints where a bounded-domain solution varies
most sharply \cite{gottlieb1977}, and the resulting differentiation operators
are well conditioned \cite{canuto2006}.

A growing literature carries that structure into neural architectures:
Chebyshev feature networks for function approximation \cite{xu2024chebfeature},
Chebyshev spectral neural networks for PDEs \cite{yin2024cheb}, spectral neural
operators combining Fourier and Chebyshev representations
\cite{fanaskov2023spectral}, and an orthogonal polynomial neural operator
\cite{liu2024opno}, which targets nonperiodic boundary conditions directly.
Most of these parametrize the solution in spectral space or map between
coefficient vectors, so the spectral representation becomes part of the model's
output as well as its input; that buys accuracy but forfeits the ability to
attribute a change in performance to the coordinate representation alone.
SEDONet \cite{abid2026} takes a minimal route: it replaces only
the trunk's coordinate input with a fixed tensor-product Chebyshev dictionary,
leaving the branch, the factorization and the parameter count untouched, and
reports consistent improvements over baseline DeepONet on a suite of canonical
benchmarks, with the largest gains on bounded, non-periodic problems. That
minimality is what makes the controlled comparison of this paper possible.

\subsection{Structural constraints, representational limits, and the remaining
gap}
\label{sec:related:structure}
The two structural facts noted above, an affine constraint on the state and
the bounded rank of a branch-trunk prediction, each carry a consequence
specific to this benchmark. On the first, what makes the present setting
unusual is not that the constraint can be imposed exactly \cite{beucler2021},
but the other side of the comparison: the reference data satisfies it to
round-off by construction, so a surrogate's violation is unambiguously its own
and can be reported as a diagnostic rather than absorbed into the data's own
error budget. Benchmarks whose reference data merely approximates a constraint
cannot make that separation, and a surrogate evaluated on them can appear
constraint-consistent for the wrong reason.

On the second, the slow decay of the Kolmogorov width for advection-dominated
problems \cite{cohen2015} has a well-documented cause: a translating sharp
feature is not well approximated by any fixed linear combination of fixed
spatial modes. The limitation is recorded in the reduced-basis literature
\cite{ohlberger2016, peherstorfer2022} and has motivated nonlinear alternatives
such as manifold-based reduction with deep convolutional autoencoders
\cite{lee2020}. The bound applies to every architecture we compare, and it is
the most likely reason all three saturate at a similar accuracy floor. It also
locates the contribution of a coordinate dictionary correctly: a dictionary
changes how quickly and stably a good basis is found, not which bases are
reachable.

Taken together, the picture is as follows. Three-phase phase-field simulation
is mature but expensive, and its structural properties, an algebraic
constraint on the state, a moving triple junction, and a genuine topological
transition are exactly the properties that operator-learning benchmarks
currently omit. Coordinate embeddings for the trunk are a promising and
minimally invasive intervention, but the evidence for them comes from
single-field targets on unconstrained domains, and it is reported as scalar
error norms that cannot say where an architectural change helps or which of
several plausible mechanisms produced it. Nothing in the existing literature
answers whether a trunk dictionary still matters on a multi-channel,
constrained, interface-dominated target. The benchmark and the controlled
comparison that follow are designed to answer both questions, and the error
decompositions of \Cref{sec:results} are designed to answer the second one in a
form that transfers to problems other than this one.

\section{Governing Model: Ternary Cahn--Hilliard--Navier--Stokes}
\label{sec:model}

This section states the model that generates the reference data and
establishes the one structural property on which the evaluation protocol
depends. \Cref{fig:sequence} shows what the model produces.
\Cref{sec:model:transport} gives the Boyer-Lapuerta ternary Cahn-Hilliard
system and the algebraic structure that keeps the phase variables on the Gibbs
simplex; \Cref{sec:model:momentum} couples it to variable-density
Navier--Stokes and fixes the boundary conditions; \Cref{sec:model:simplex}
shows that the simplex constraint is an identity of the model rather than an
approximation, which is what makes any simplex violation by a surrogate
attributable to the surrogate; and \Cref{sec:nondim} reduces the physics to the
nine dimensionless parameters that the ensemble varies.

\begin{figure}[H]
\centering
\includegraphics[width=0.82\textwidth]{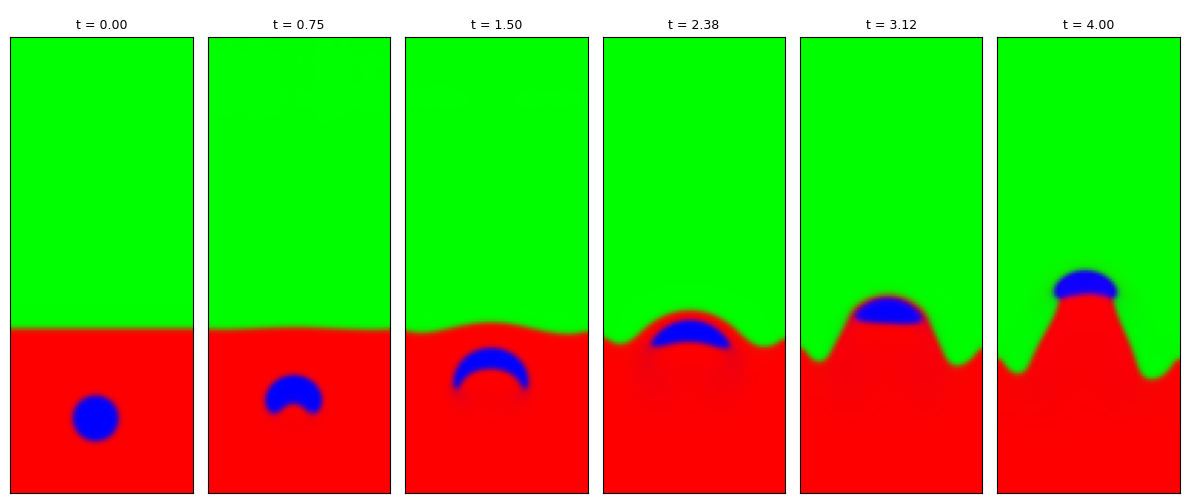}
\caption{The target physics: RGB composite of the three volume fractions (red
$=$ water, green $=$ oil, blue $=$ air) at
$t = 0,\,0.75,\,1.50,\,2.38,\,3.12,\,4.00$. The bubble rises, deforms into a
spherical cap, pierces the water-oil interface, and entrains a water plume
that thins and pinches as the bubble detaches.}
\label{fig:sequence}
\end{figure}

\subsection{Phase transport}
\label{sec:model:transport}

Let $\Omega = (0,L_x)\times(0,L_y) \subset \R^2$ be a bounded, wall-confined
domain and let $\vc = (c_1,c_2,c_3)$ denote the volume fractions of water,
oil, and air, for which we write $(c_w,c_o,c_a) \equiv (c_1,c_2,c_3)$
interchangeably. The three fields are not independent: at every point of space
and time they are constrained to the Gibbs simplex,
\begin{equation}
\textstyle\sum_{i=1}^{3} c_i(\vx,t) = 1,
\qquad c_i \ge 0 .
\label{eq:simplex}
\end{equation}
This algebraic constraint has no counterpart in the two-phase problem, where a
single order parameter carries all the information, and it is the first of the
two features that make the ternary system a qualitatively different learning
target. \Cref{sec:model:simplex} concerns how \cref{eq:simplex} is maintained
by the model; how a surrogate should be made to respect it is taken up later.

The second such feature is that the three pairwise surface tensions
$\sigma_{ij}$ act jointly rather than separately. They enter the model through
the spreading coefficients \cite{boyer2006}
\begin{equation}
\Sigma_i = \sigma_{ij} + \sigma_{ik} - \sigma_{jk},
\qquad
\frac{3}{\Sigma_T} = \frac{1}{\Sigma_1}+\frac{1}{\Sigma_2}+\frac{1}{\Sigma_3},
\label{eq:spreading}
\end{equation}
with $\{i,j,k\}$ a permutation of $\{1,2,3\}$ and $\Sigma_T$ the harmonic mean
that will reappear below as a Lagrange multiplier. Their signs decide the
wetting regime. If $\Sigma_i>0$ for all $i$, the configuration is
partial-spreading: the three phases meet at a triple junction whose contact
angles are fixed jointly by the three tensions through a Neumann-triangle
balance \cite{degennes1985}. If $\Sigma_i<0$ for some $i$, the configuration is
total-spreading: phase $i$ films spontaneously between the other two and the
triple junction ceases to exist. The second regime is physically real but is
not a well-posed target for a single-valued solution operator, since
arbitrarily close parameter values then produce topologically distinct
arrangements; the sampling box of \Cref{sec:nondim} is accordingly
parametrized so that $\Sigma_i > 0$ holds identically over it. Those coefficients set the depth of the wells in the bulk free energy density,
which is the triple-well
\begin{equation}
F_0(\vc) = \sum_{i=1}^{3}\frac{\Sigma_i}{2}\,c_i^2(1-c_i)^2
 \;+\; \Lambda\, c_1^2 c_2^2 c_3^2 ,
\label{eq:F0}
\end{equation}
whose three minima at the pure-phase vertices of the simplex are what make the
diffuse interfaces stable. The optional term with $\Lambda \ge 0$ penalizes
spurious formation of a third phase at what should be a two-phase interface,
but throughout the present ensemble $\Lambda = 0$: over the admissible
spreading coefficients the pure triple-well is already free of spurious
third-phase layers, and the derivative $2\Lambda\prod_j c_j^2/c_i$ is singular
wherever a phase vanishes, which is almost everywhere in this problem.

Given the free energy, the phase fields are advected by the flow and relax by
Cahn--Hilliard dynamics,
\begin{equation}
\partial_t c_i + \nabla\!\cdot(\vu\,c_i) \;=\; \nabla\!\cdot\!\big(M_i \nabla \mu_i\big),
\qquad i = 1,2,3 ,
\label{eq:ch}
\end{equation}
with generalized chemical potentials
\begin{equation}
\mu_i \;=\; \underbrace{\frac{12}{\varepsilon}\,\frac{\partial F_0}{\partial c_i}
\;-\; \frac{4\,\Sigma_T}{\varepsilon}\sum_{j=1}^{3}\frac{1}{\Sigma_j}\,
      \frac{\partial F_0}{\partial c_j}}_{\textstyle =:\,N_i(\vc)}
\;-\; \frac{3}{4}\,\varepsilon\,\Sigma_i\,\nabla^2 c_i ,
\label{eq:mu}
\end{equation}
in which $\varepsilon$ is the interface thickness parameter and $M_i$ the
mobility of species $i$. We name the algebraic, gradient-free part $N_i$
explicitly because it carries the cancellation of \Cref{prop:simplex}. The
remaining term is the surface term, whose coefficient
$\tfrac34\varepsilon\Sigma_i$ sets both the equilibrium profile width and the
tension of interface $i$. The second term of \cref{eq:mu} deserves particular attention, because it is
independent of $i$. It is therefore not a species-specific contribution to the
thermodynamics but a Lagrange multiplier enforcing \cref{eq:simplex}, and its
prefactor $4\Sigma_T/\varepsilon$ is not a matter of convention. Together with
the mobility choice $M_i = M_0/\Sigma_i$ made in \Cref{sec:model:simplex}, that
prefactor is exactly what causes the weighted sum $\sum_i M_i N_i$ to vanish
identically. This is the single most important structural feature of the model
for our purposes, and everything the benchmark claims about constraint
preservation rests on it.

\subsection{Momentum and incompressibility}
\label{sec:model:momentum}

The phase fields of \cref{eq:ch} are transported by a velocity field that they
in turn determine, since the three fluids are treated as a single
incompressible mixture whose density and viscosity are composition-dependent,
\begin{equation}
\rho(\vc) = \sum_{i=1}^{3}\rho_i c_i,
\qquad
\eta(\vc) = \sum_{i=1}^{3}\eta_i c_i ,
\label{eq:props}
\end{equation}
so that both vary smoothly across each diffuse interface between their bulk
values. In the present ensemble the density ratio across the water-air
interface reaches $100$ and the viscosity ratio across the water-oil
interface reaches $3$, so the variable-coefficient structure that
\cref{eq:props} introduces into the momentum equation is not a formality. With those properties fixed, the mixture obeys
\begin{align}
\rho(\vc)\big(\partial_t \vu + \vu\!\cdot\!\nabla \vu\big)
&= -\nabla p
 + \nabla\!\cdot\!\Big[\eta(\vc)\big(\nabla\vu + \nabla\vu^{\!\top}\big)\Big]
 - \sum_{i=1}^{3} c_i \nabla \mu_i
 - \rho(\vc)\, g\, \bs{e}_y ,
\label{eq:mom}\\[2pt]
\nabla\!\cdot\vu &= 0 ,
\label{eq:div}
\end{align}
in which the buoyancy that drives the entire problem enters through the last
term of \cref{eq:mom}, as the imbalance between the local mixture density and
its surroundings. There is no separate forcing and no imposed inflow; the
bubble rises because \cref{eq:props} makes it lighter than the water around
it.

The capillary force appears as the potential term $-\sum_i c_i \nabla\mu_i$,
the ternary generalization of the continuum-surface-force representation
\cite{brackbill1992}, and this form is what makes a three-phase computation
tractable at all. A curvature-based force would require an estimate of the
interface normal and curvature for each of the three interfaces and, at the
triple junction, an explicit model of the junction geometry, precisely the
machinery that sharp-interface methods must supply and that fails when the
topology changes. The potential form requires none of it. It is assembled from
quantities the Cahn-Hilliard system already computes, and it produces the
correct Neumann-triangle balance at the junction automatically, as a
consequence of \cref{eq:spreading} rather than as an imposed condition. It remains to close the system at the walls. The domain is a closed box: the
phase fields and chemical potentials satisfy homogeneous Neumann conditions,
corresponding to neutral wetting; the velocity satisfies no-slip and
no-penetration; and the pressure satisfies homogeneous Neumann,
\begin{equation}
\partial_n c_i = \partial_n \mu_i = 0,\quad
\vu = \bs{0},\quad
\partial_n p = 0
\qquad \text{on } \partial\Omega .
\label{eq:bc}
\end{equation}
This choice matters in two ways. Every condition on a scalar in
\cref{eq:bc} is homogeneous Neumann, which permits an exact cosine
diagonalization of the discrete Laplacian and hence a transform-based solve at
every step; and the domain is strictly non-periodic in both $x$ and $y$, a
property that \Cref{sec:embeddings} returns to when the three trunk
embeddings are compared.
\subsection{Structural consistency of the simplex constraint}
\label{sec:model:simplex}

The mobilities have so far been left free. Fixing them in inverse proportion
to the spreading coefficients turns \cref{eq:simplex} from a constraint that
must be enforced into an identity that holds automatically.

\begin{proposition}[Simplex preservation]
\label{prop:simplex}
Let the mobilities be chosen as $M_i = M_0/\Sigma_i$ with $M_0>0$ constant.
Then the algebraic part of the chemical potential satisfies
\begin{equation}
\sum_{i=1}^{3} M_i\,N_i(\vc) \;\equiv\; 0
\qquad \text{for every } \vc \in \R^3 ,
\label{eq:consistency}
\end{equation}
and, if in addition $\vu$ is divergence free and $\sum_i c_i = 1$, then
$\sum_i M_i \mu_i \equiv 0$ and
$\partial_t\!\left(\sum_i c_i\right) = 0$; the constraint \cref{eq:simplex}
is propagated exactly by \cref{eq:ch}.
\end{proposition}

The proof is given in \Cref{app:proof}. Two aspects of it are worth isolating.
First, \cref{eq:consistency} is unconditional: it holds pointwise for
arbitrary $\vc$, whether or not the simplex constraint is currently satisfied,
because it follows from the identity $\sum_i \Sigma_i^{-1} = 3/\Sigma_T$ alone.
Second, the surface term of \cref{eq:mu} contributes
$-\tfrac{3}{4}\varepsilon M_0 \nabla^2 c_i$ once $M_i\Sigma_i = M_0$, so its
coefficient is species-independent. The cancellation therefore depends neither
on the solution being smooth nor on the constraint holding beforehand, which is
what allows a discretization to inherit it
(\Cref{prop:simplex-discrete}).

The practical content is that $\sum_i c_i = 1$ is a structural identity of the
model and of its discretization, not a condition the solver is nudged toward by
penalty or projection. In the reference ensemble the measured drift is
$\max_{\vx,t}\big|\sum_i c_i - 1\big| = 2.1\times10^{-4}$, at the floor of
single-precision accumulation. Because the training data satisfies the
constraint to round-off, any simplex violation exhibited by a surrogate trained
on that data is provably its own and is therefore an unambiguous diagnostic of
the learned model. Benchmarks whose reference data merely approximates a
constraint cannot make this separation.

\subsection{Nondimensionalization and parameters}
\label{sec:nondim}

\Cref{eq:ch,eq:mom,eq:bc} specify the model; what remains is to reduce it to
the finite set of numbers the ensemble varies. All quantities are reported in
computational units in which the channel width, the water density and gravity
are the references, $L_x = 1$, $\rho_w = 1$ and $g = 0.98$, so that
$\Omega = (0,1)\times(0,2.5)$ and the horizon is $t\in[0,4]$, with the bubble
diameter fixed at $d = 0.25\,L_x$ and the undisturbed water--oil interface at
$y_{\mathrm{int}} = 0.9\,L_x$. The velocity scale is built on the bubble
rather than the channel, since it is the bubble that sets every time scale in
the problem,
\begin{equation}
U = \sqrt{g\,d} = 0.4950 ,
\label{eq:Uscale}
\end{equation}
so that $t = 4$ amounts to $\approx 7.9$ convective times $d/U$, long enough
for rise, breakthrough, entrainment and detachment, and short enough that no
member reaches the lid.

Sampling the three tensions independently would waste much of the box on
inadmissible triples and would distort the design near the total-wetting
boundary, so we instead set the water-air tension by an E\"otv\"os number
built on $d$ and the other two by dimensionless ratios,
\begin{equation}
\mathrm{Eo} = \frac{\rho_w\, g\, d^2}{\sigma_{aw}},
\qquad
\sigma_{ow} = r_{ow}\,\sigma_{aw},
\qquad
\sigma_{ao} = \big(1 + s - r_{ow}\big)\,\sigma_{aw} ,
\label{eq:tensions}
\end{equation}
which on substitution into \cref{eq:spreading} gives the spreading
coefficients in closed form,
\begin{equation}
\frac{\Sigma_w}{\sigma_{aw}} = 2r_{ow} - s ,
\qquad
\frac{\Sigma_o}{\sigma_{aw}} = s ,
\qquad
\frac{\Sigma_a}{\sigma_{aw}} = 2 + s - 2r_{ow} .
\label{eq:Sigma-closed}
\end{equation}
This exposes the meaning of the third parameter, since $s$ \emph{is} the
normalized oil spreading coefficient, with $s\to0^{+}$ the total-wetting limit
and larger $s$ pushing the triple junction toward a symmetric Neumann
triangle. Over the sampling box the three ratios lie in $[0.10,1.38]$,
$[0.02,0.50]$ and $[0.62,1.90]$, all strictly positive, so every design point
is admissible and no rejection is required, and with $g$ and $d$ fixed the
range $\mathrm{Eo}\in[5,40]$ corresponds to
$\sigma_{aw} \in [1.53\times10^{-3},\,1.23\times10^{-2}]$.

Viscosities are set on the same scales by a Reynolds number, with the oil
viscosity a free ratio and the air viscosity fixed by a constant kinematic
ratio,
\begin{equation}
\mathrm{Re} = \frac{\rho_w U d}{\eta_w},
\qquad
\eta_o = (\eta_o/\eta_w)\,\eta_w ,
\qquad
\frac{\nu_a}{\nu_w} = 2 \;\;(\text{fixed}) ,
\label{eq:viscosities}
\end{equation}
where tying $\eta_a$ to the air density reflects a physical link and avoids
tightening the explicit viscous time-step limit without adding physics; the
range $\mathrm{Re}\in[20,70]$ corresponds to
$\eta_w \in [1.77\times10^{-3},\,6.19\times10^{-3}]$. One group is
deliberately held fixed: the Cahn number
$\mathrm{Cn} = \varepsilon/L_x = 3.125\times10^{-2}$ gives
$\varepsilon/d = 0.125$, so the diffuse layer is one eighth of a bubble
diameter, thin enough to resolve the physics and thick enough to be resolved
by the grid, and varying it would make the ensemble a study of discretization
error as much as of physics. The nine quantities actually varied are
$\mathrm{Eo}$, $\mathrm{Re}$, $s$, $r_{ow}$, the oil-to-water and air-to-water
density ratios, the oil-to-water viscosity ratio, and the two coordinates of
the initial bubble centre.

\section{Operator-Learning Formulation}
\label{sec:opnet}

This section sets up the learning problem and then isolates the one degree of
freedom that the paper actually varies. \Cref{sec:opnet:operator} defines the
target operator and records the two structural properties, a constrained
codomain and a shared underlying geometry, that distinguish it from the
single-field benchmarks on which trunk embeddings have previously been tested.
\Cref{sec:opnet:factorization} gives the branch-trunk factorization for a
five-channel output, states what a rank-$p$ factorization can and cannot
represent, and imposes the simplex constraint structurally.
\Cref{sec:embeddings} defines the three trunk embeddings, which constitute the
entire difference between DeepONet, FEDONet and SEDONet, and
\Cref{sec:opnet:compare} explains in terms of feature conditioning why that
difference should be expected to matter at all. The remaining subsections give
the architectures and their parameter accounting, the training protocol, and
one training algorithm per method.

\subsection{The parametric solution operator}
\label{sec:opnet:operator}

Let $\vp \in \mathcal{P}\subset\R^{9}$ be the vector of physical parameters of
\Cref{sec:nondim}, so that $\mathcal{P}$ is the compact nine-dimensional box
defined by their sampling ranges, and let
$\vzeta = (x,y,t) \in \Omega\times[0,T]$ with
$\Omega=(0,L_x)\times(0,L_y)$, $L_x=1$, $L_y=2.5$ and $T=4$. The map we learn
is the parametric solution operator
\begin{equation}
\Gop:\;\vp \;\longmapsto\;
\bs{s}(\cdot\,;\vp) = \big(c_w, c_o, c_a, u, v\big)(\vzeta;\vp),
\qquad
\Gop: \mathcal{P} \to \mathcal{S}, \quad
\mathcal{S} \subset C\big(\overline\Omega\times[0,T];\,\R^{5}\big),
\label{eq:operator}
\end{equation}
whose codomain is a space of functions rather than a finite-dimensional vector
space: $\Gop$ sends each of the nine numbers to an entire space--time field,
and it is this that makes the problem operator learning rather than
regression. What the surrogate must learn is not a solution but the map that
produces solutions.

For \cref{eq:operator} to be a legitimate learning target, $\Gop$ must be
single-valued and continuous on $\mathcal{P}$. Single-valuedness is not
automatic in multiphase flow, since a topological bifurcation, a member in
which the bubble fails to cross, or in which the plume pinches instead of
persisting, would make arbitrarily close parameters produce qualitatively
different fields, and no continuous operator could interpolate between them.
Diagnostics computed over the ensemble rule this out: every member follows the
same qualitative sequence, and the scalar responses vary smoothly and
monotonically with the dominant parameter. Within the sampled box $\Gop$ is
therefore a continuous map from a compact set into
$C(\overline\Omega\times[0,T];\R^5)$, and the operator universal approximation
theorem \cite{chenchen1995}, in its branch-trunk form \cite{lu2021},
guarantees that for any $\epsilon>0$ there exist networks and a latent width
$p$ achieving
$\sup_{\vp\in\mathcal{P}}\|\Gop(\vp) - \widehat{\Gop}(\vp)\|_\infty < \epsilon$.
The theorem is an existence statement and says nothing about $p$, about the
number of samples, or about whether gradient descent will find the
approximation; the entire empirical content of this paper concerns the last of
these.

Two further properties of $\mathcal{S}$ are worth recording before any
architecture is chosen, because both have architectural consequences. The
first is that the image of $\Gop$ does not fill $\R^5$ pointwise: by
\cref{eq:simplex} the first three components lie on the two-dimensional affine
slice $\Pi = \{\bs{c}\in\R^3: c_w+c_o+c_a=1\}$, so the effective pointwise
output dimension is four rather than five, and any surrogate whose codomain is
all of $\R^5$ is by construction able to leave the physically admissible set.
The second is that the five components are not independent observables of
independent physics but five functionals of one moving geometry: the same
interface that separates $c_w$ from $c_o$ carries the capillary jump that
structures $v$, and the same bubble rim that gives $c_a$ its high-wavenumber
content drives the recirculation that appears in $u$.
\Cref{sec:opnet:factorization} takes up both.

One specialization of the general setting \cite{lu2021,kovachki2023} is worth
noting, because it sharpens the test. There the branch consumes sensor samples
$[u_0(x_1),\dots,u_0(x_m)]$ of an input function $u_0$, and the total
error decomposes into the finite-sensor discretization of $u_0$, the
representation error of the branch--trunk product, and the optimization error.
Here the input function space is the nine-dimensional parametric family, so
the branch consumes $\vp$ directly with $m=9$ and the first term is closed
exactly: the nine numbers determine the simulation bit for bit. Whatever error
is measured in \Cref{sec:results} is therefore representation error plus
optimization error, and a controlled change in the trunk acts on it directly
rather than competing with a sensor-resolution term. Nothing else in the
factorization depends on the distinction, the branch reads a
finite-dimensional encoding exactly as it would read a sensor vector, and the
trunk, where all three architectures differ, is untouched by it.

\subsection{Branch-Trunk factorization and the simplex closure}
\label{sec:opnet:factorization}

We use a shared trunk basis with per-channel branch coefficients. The branch
network $\Bnet$ maps $\vp$ to $5p$ latent coefficients, reshaped as
$b^{(c)}_k(\vp)$ for channel $c\in\{1,\dots,5\}$ and mode $k\in\{1,\dots,p\}$,
while the trunk network $\Tnet$ maps the embedded coordinate $\phi(\vzeta)$ to
$p$ latent basis channels $t_k$, so that the prediction is
\begin{equation}
\widehat{s}_c(\vzeta;\vp)
\;=\; \sum_{k=1}^{p} b^{(c)}_{k}(\vp)\; t_{k}\big(\phi(\vzeta)\big)
\;+\; b^{(c)}_{0} ,
\qquad
t_k = \big[\Tnet\circ\phi\big]_k ,
\quad
b^{(c)}_k = \big[\Bnet(\vp)\big]_{(c,k)} ,
\label{eq:donet}
\end{equation}
with $b^{(c)}_0$ a per-channel bias \cite{lu2021}. The two networks never
interact except through the sum in \cref{eq:donet}: the branch never sees a
coordinate and the trunk never sees a parameter. This separation is what makes
the architecture an operator approximator rather than a function approximator
on $\mathcal{P}\times\Omega\times[0,T]$, and it is also what makes the trunk a
clean experimental target, since any change to $\phi$ alters the coordinate
representation without touching the parameter encoding. Written for a batch of query points the structure becomes explicit. Let
$Z = \{\vzeta^{q}\}_{q=1}^{Q}$ be the query grid, let
$\bs{T}(Z)\in\R^{Q\times p}$ have entries $T_{qk} = t_k(\phi(\vzeta^q))$, and
let $\bs{B}(\vp)\in\R^{5\times p}$ collect the branch coefficients; then
\begin{equation}
\widehat{\bs{S}}(\vp) \;=\; \bs{T}(Z)\,\bs{B}(\vp)^{\!\top}
\;+\; \bs{1}_Q\, \bs{b}_0^{\!\top}
\;\in\; \R^{Q\times 5} ,
\label{eq:donet-matrix}
\end{equation}
a separated representation of rank at most $p$ in which the $Q\times p$ factor
is shared across every parameter value and every channel and the $5\times p$
factor carries all parameter dependence. This is the structure of a
reduced-basis or proper-orthogonal-decomposition expansion, with the single
difference that the basis $\bs{T}(Z)$ is learned by gradient descent rather
than computed from a snapshot matrix.

That observation carries a genuine limitation, and it is worth stating
explicitly because it bounds what any DeepONet variant can achieve here. Since
\cref{eq:donet-matrix} is linear in the trunk basis, the best achievable error
is bounded below by the error of the best $p$-dimensional linear subspace of
the solution manifold, the Kolmogorov $p$-width. Writing
$\bs{S}\in\R^{Q\times 5N}$ for the snapshot matrix of the training ensemble
with singular values $\varsigma_1\ge\varsigma_2\ge\cdots$, the best rank-$p$
approximation in the Frobenius norm satisfies
\begin{equation}
\min_{\mathrm{rank}(\widehat{\bs S})\le p}
\big\|\bs{S}-\widehat{\bs{S}}\big\|_F^2 = \sum_{j>p}\varsigma_j^2 ,
\label{eq:eckart}
\end{equation}
so no choice of $\phi$, and no amount of training, can push the error of a
rank-$p$ model below that tail. For advection-dominated problems with moving
interfaces the singular values decay slowly, because a translating sharp
feature is not well approximated by any fixed linear combination of fixed
spatial modes, and this is the standard obstruction to linear reduced-order
modeling of transport. It is the most likely explanation for the accuracy
floor all three models reach in \Cref{sec:results}, and it locates the trunk
embedding correctly: $\phi$ does not change what $\bs{T}$ is capable of
representing, only how quickly and stably gradient descent finds a good
$\bs{T}$. Breaking that barrier would require a nonlinear parametrization of
the basis, which we note as future work in \Cref{sec:conclusion} rather than
pursue here.

Within that bound, a single trunk is shared by all five channels deliberately.
One trunk per channel would multiply the trunk parameters by five and would
force each channel to rediscover the same interface geometry from its own
supervision alone. Under \cref{eq:donet-matrix} the five channels instead
share one set of spatiotemporal basis functions and differ only in the
coefficients with which they combine them, which is the correct encoding of
the second property recorded above. The price is that a basis good for one
channel must be good for all, so a dictionary that helped the phase fields at
the expense of the velocities would show up immediately; the per-channel
breakdown of \Cref{tab:test} is the test of exactly this.

The first property is handled at the output rather than in the basis. Because
\cref{eq:simplex} holds exactly in the data (\Cref{prop:simplex}), it would be
perverse to let the surrogate violate it, and since the constraint is affine it
can be imposed structurally at no cost. We predict two concentration channels
and close the third,
\begin{equation}
\widehat{c}_a \;=\; 1 - \widehat{c}_w - \widehat{c}_o ,
\label{eq:closure}
\end{equation}
a section of the constraint map that parametrizes $\Pi$ by two free
coordinates; the alternative, predicting all three and projecting orthogonally
onto $\Pi$, distributes the residual equally and is the minimum-norm
correction. Either way the constraint holds to floating-point rounding, so the
simplex violation we report is a machine-epsilon quantity rather than a
modelling error. One implementation point is worth stating because getting it
wrong fails silently: since each channel is standardized to zero mean and unit
variance for training, \cref{eq:closure} is
not the identity
$\widehat{c}_a^{\,\mathrm{std}} = 1 - \widehat{c}_w^{\,\mathrm{std}} -
\widehat{c}_o^{\,\mathrm{std}}$ in standardized variables. Writing
$c_i = \mu_i + s_i c_i^{\mathrm{std}}$ with per-channel mean $\mu_i$ and scale
$s_i$, the closure becomes
\begin{equation}
\widehat{c}_a^{\,\mathrm{std}}
= \frac{1-\mu_w-\mu_o-\mu_a}{s_a}
\;-\; \frac{s_w}{s_a}\,\widehat{c}_w^{\,\mathrm{std}}
\;-\; \frac{s_o}{s_a}\,\widehat{c}_o^{\,\mathrm{std}} ,
\label{eq:closure-std}
\end{equation}
an affine map with channel-specific coefficients; the unweighted form leaves an
$\mathcal{O}(1)$ residual while looking entirely correct in code.

The closure also reshapes the supervision, and it is worth working out how.
The loss is still evaluated on all three concentration channels, so no phase is
privileged, but the three errors are no longer independent. Let
$e_w = \widehat c_w - c_w$ and $e_o = \widehat c_o - c_o$; under
\cref{eq:closure}, and using $\sum_i c_i = 1$ in the data, the third error is
determined, $e_a = -(e_w+e_o)$, and the three-channel squared error collapses
to a quadratic form in two variables,
\begin{equation}
e_w^2 + e_o^2 + e_a^2
= 2\,\bs{e}^{\!\top} A\, \bs{e},
\qquad
\bs{e} = (e_w,e_o)^{\!\top},
\qquad
A = \begin{pmatrix} 1 & \tfrac12 \\[2pt] \tfrac12 & 1 \end{pmatrix} ,
\label{eq:closure-metric}
\end{equation}
with eigenvalues $\tfrac32$ along $(1,1)/\sqrt2$ and $\tfrac12$ along
$(1,-1)/\sqrt2$. The closure therefore penalizes co-directional errors in water
and oil, which are exactly errors in the total liquid fraction, that is, in
the placement of the air interface, three times more heavily than exchange
errors, in which water is mistaken for oil at fixed total liquid. For the
present problem this is a favorable reweighting, since the bubble boundary is
the hardest and most consequential feature in the field, and the closure
upweights it without a hand-tuned channel weight; the induced metric has
condition number $3$, small enough that no optimization penalty accompanies it.
The asymmetry is a property of the section rather than of the constraint: under
the orthogonal projection, the penalty would be isotropic on $\Pi$.

\subsection{Trunk embeddings: the only difference between the three models}
\label{sec:embeddings}

All three architectures use \cref{eq:donet} verbatim, with identical branch
network, identical trunk depth and width, identical latent width $p$, identical
loss, optimizer, schedule and data. They differ only in the fixed,
non-trainable map $\phi$ that preprocesses the query coordinate before it
reaches the first trunk layer. This is the experimental control on which the
entire paper rests, so we state each $\phi$ explicitly and in full. Every embedding acts on affinely rescaled coordinates rather than on physical
ones, so that the three axes, whose physical extents $L_x=1$, $L_y=2.5$ and
$T=4$ differ by a factor of four, are treated comparably:
\begin{equation}
\xi_x = \frac{2x}{L_x}-1,\qquad
\xi_y = \frac{2y}{L_y}-1,\qquad
\xi_t = \frac{2t}{T}-1,
\qquad \bs{\xi}=(\xi_x,\xi_y,\xi_t) \in [-1,1]^3 .
\label{eq:affine}
\end{equation}
For the Chebyshev dictionary this is a requirement rather than a convention,
since $\{T_n\}$ is orthogonal on $[-1,1]$ and nowhere else. For the Fourier
features it is a fairness condition: applying a single isotropic bandwidth to
unscaled $(x,y,t)$ would give the temporal axis four times the effective
frequency content of the horizontal one, so the two embeddings would no longer
be compared on equal footing.

\subsubsection{DeepONet: Raw coordinates}
The baseline applies no embedding at all,
\begin{equation}
\phi_{\mathrm{raw}}(\vzeta) = \bs{\xi} \in \R^{3},
\qquad d_{\mathrm{trunk}} = 3 ,
\label{eq:phi-raw}
\end{equation}
so that every basis function in $\bs{T}(Z)$, every interface profile, every
oscillation, every sharp gradient, must be synthesized from compositions of
affine maps and pointwise nonlinearities acting on three numbers. This is
exactly the configuration to which the spectral-bias analysis of
\cite{rahaman2019,xu2020frequency,jacot2018ntk} applies: the neural tangent
kernel of a shallow MLP on raw coordinates has an eigenspectrum dominated by
low-frequency eigenfunctions, so high-wavenumber content is learned last, if at
all. For a phase-field target this is the worst possible allocation of prior
mass, since the bulk of the solution is nearly constant and carries almost no
information while everything of interest lives in the thin diffuse layers.

\subsubsection{FEDONet: Random Fourier features}

The first remedy lifts the coordinate into a sinusoidal dictionary with random
frequencies \cite{sojitra2025,tancik2020},
\begin{equation}
\phi_{\mathrm{F}}(\vzeta)
= \Big[\sin\!\big(2\pi B\bs{\xi}\big),\; \cos\!\big(2\pi B\bs{\xi}\big)\Big]
\in \R^{2M},
\qquad
B \in \R^{M\times 3},\quad B_{ij}\sim\mathcal{N}\!\big(0,\sigma_B^2\big),
\label{eq:phi-fourier}
\end{equation}
with $M = 128$, hence $d_{\mathrm{trunk}} = 2M = 256$. The matrix $B$ acts on
the normalized coordinates of \cref{eq:affine}, is drawn once at a fixed
bandwidth and frozen thereafter: it is not trained and does not appear in
$\theta$, so the embedding contributes no parameters and the comparison with
SEDONet remains exact. The mechanism is best seen through the kernel that \cref{eq:phi-fourier}
induces. By Bochner's theorem and the random-features construction
\cite{rahimi2007}, the inner product of two embedded coordinates is a Monte
Carlo estimate of the Fourier transform of the sampling density,
\begin{equation}
\frac{1}{M}\,\phi_{\mathrm{F}}(\vzeta)^{\!\top}\phi_{\mathrm{F}}(\vzeta')
= \frac{1}{M}\sum_{j=1}^{M}\cos\!\big(2\pi\, \bs{b}_j^{\!\top}(\bs{\xi}-\bs{\xi}')\big)
\;\xrightarrow[M\to\infty]{}\;
\exp\!\Big(-2\pi^2\sigma_B^2\,\big\|\bs{\xi}-\bs{\xi}'\big\|^2\Big) ,
\label{eq:rff-kernel}
\end{equation}
a Gaussian kernel of length scale $(2\pi\sigma_B)^{-1}$. Three properties
follow that matter here. The kernel is stationary, depending on
$\bs{\xi}-\bs{\xi}'$ alone, so the prior is blind to the boundary of the box
and is identical at the walls and in the interior. It is isotropic at
fixed $\sigma_B$, so one scalar sets the resolvable scale in all three
coordinates, a strong assumption for a field whose structure in $t$ is a
monotone drift and whose structure in $x$ is a thin neck. And it is built from
functions periodic on the embedding scale, so a non-periodic component of the
target, such as a linear ramp in $\xi_t$, must be represented by a slowly
converging sinusoidal series with Gibbs oscillations at the endpoints.

\subsubsection{SEDONet: Chebyshev dictionary}

The second remedy replaces that sinusoidal prior with a polynomial one. Let
$\{T_n\}_{n\ge0}$ be the Chebyshev polynomials of the first kind, generated by
the three-term recurrence
\begin{equation}
T_0(\xi) = 1,\quad T_1(\xi)=\xi,\quad
T_{n+1}(\xi) = 2\xi\,T_n(\xi) - T_{n-1}(\xi),
\label{eq:cheb-rec}
\end{equation}
orthogonal on $[-1,1]$ under the weight $w(\xi)=(1-\xi^2)^{-1/2}$ and
satisfying $T_n(\cos\vartheta)=\cos(n\vartheta)$. That last identity is worth
keeping in view: a Chebyshev expansion is a cosine series in a stretched
variable, so the dictionary is not an alternative to a Fourier basis so much as
a Fourier basis on a coordinate that clusters near the endpoints and imposes no
periodicity. The three-dimensional tensor-product dictionary indexed by the
multi-index $\bs{\alpha}=(k,l,m)$ is
\begin{equation}
\Phi_{\bs{\alpha}}(\vzeta) \;=\; T_k(\xi_x)\,T_l(\xi_y)\,T_m(\xi_t),
\qquad 0 \le k,l,m < K ,
\label{eq:cheb-tensor}
\end{equation}
and the embedding is a truncation of these $K^3$ modes to an index set
$\mathcal{A}$ of cardinality $d_{\mathrm{trunk}}$,
\begin{equation}
\phi_{\mathrm{C}}(\vzeta) = \big[\,\Phi_{\bs{\alpha}}(\vzeta)\,\big]_{\bs{\alpha}\in\mathcal{A}}
\in \R^{d_{\mathrm{trunk}}} .
\label{eq:cheb-embed}
\end{equation}
Like $B$ in \cref{eq:phi-fourier}, $\mathcal{A}$ is fixed before training and
carries no parameters; unlike $B$, it involves no random draw, so
$\phi_{\mathrm{C}}$ is bit-identical across seeds and requires no bandwidth
selection. The choice of $\mathcal{A}$ deserves attention, because the obvious
choice fails in three coordinates. In one spatial dimension plus time
\cite{abid2026} uses a lexicographic crop, which is harmless when almost every
mode is retained. In three coordinates it is not: cropping the
lexicographically ordered $K^3$ modes to the first
$d_{\mathrm{trunk}} \ll K^3$ entries retains
\begin{equation}
\mathcal{A}_{\mathrm{lex}}
= \Big\{(k,l,m) : k < \big\lceil d_{\mathrm{trunk}}/K^2\big\rceil\Big\}
\quad\Longrightarrow\quad
k \in \{0,1,2\} \ \text{ for } K=10,\; d_{\mathrm{trunk}}=256 ,
\label{eq:lex}
\end{equation}
that is, full degree in the last index and degree at most two in the first.
Such a dictionary is nearly blind to $x$-structure above quadratic order, which
is fatal on a problem whose hardest feature, the plume neck, is a thin
structure in $x$. We therefore order modes by total degree,
\begin{equation}
\mathcal{A}_{\mathrm{graded}}:\quad
\Phi_{\bs{\alpha}} \ \text{retained iff} \
\mathrm{rank}_{\preceq}(\bs{\alpha}) < d_{\mathrm{trunk}},
\qquad
\bs{\alpha} \preceq \bs{\alpha}' \iff |\bs{\alpha}| < |\bs{\alpha}'| ,
\quad |\bs{\alpha}| = k+l+m ,
\label{eq:graded}
\end{equation}
with ties broken lexicographically and the bound $\max(k,l,m)<K$ applied first.
Since the number of multi-indices of total degree at most $D$ in three
variables is
\begin{equation}
\#\big\{\bs{\alpha}\in\mathbb{Z}_{\ge0}^3 : |\bs{\alpha}| \le D\big\}
= \binom{D+3}{3},
\qquad
\#\big\{\bs{\alpha} : |\bs{\alpha}| = D\big\} = \binom{D+2}{2} ,
\label{eq:graded-count}
\end{equation}
the budget $d_{\mathrm{trunk}}=256$ with $K=10$ retains every mode of total
degree at most $9$, all $\binom{12}{3}=220$ of them, plus $36$ modes of degree
exactly $10$. The resulting dictionary is degree-balanced across $x$, $y$ and
$t$ by construction, and \cref{eq:graded-count} makes the feature budget easy
to choose in advance.
\subsection{Why the embedding changes the optimization}
\label{sec:opnet:compare}
Each $\phi$ encodes a prior about what solution fields look like, and the three
priors are genuinely different rather than three parametrizations of one thing.
The cleanest way to see why they should affect training, rather than only
expressivity, is to examine the second-order statistics of the features the
first trunk layer receives. Define the empirical feature Gram matrix over the
query set,
\begin{equation}
G \;=\; \frac{1}{Q}\sum_{q=1}^{Q}
    \phi(\vzeta^{q})\,\phi(\vzeta^{q})^{\!\top}
\;\in\; \R^{d_{\mathrm{trunk}}\times d_{\mathrm{trunk}}} .
\label{eq:gram}
\end{equation}
For the idealized linear model $\widehat s = \bs{w}^{\!\top}\phi(\vzeta)$
trained by gradient descent with step $\eta$ on a squared loss, the error
component along the $i$th eigenvector of $G$ contracts by a factor
$(1-\eta\lambda_i)$ per step, so directions with small $\lambda_i$ are learned
last and the steps needed to resolve a direction scale as $\lambda_i^{-1}$. The
conditioning of $G$ \emph{is} the learning-rate spectrum, and spectral bias is
the statement that for a raw-coordinate MLP the relevant $\lambda_i$ are tiny.
The trunk is a nonlinear network rather than a linear model, so this is a
heuristic rather than a theorem, but it is the same heuristic that underlies
the neural tangent kernel analyses of
\cite{jacot2018ntk,wang2021eigenvector} and it predicts the observed ordering
correctly.

Under this lens the three embeddings separate sharply. The raw trunk offers
almost nothing for $G$ to condition, so every useful feature direction must be
manufactured by the network itself, precisely the regime in which the NTK
eigenspectrum decays fastest. The random Fourier embedding whitens $G$ in
expectation, since distinct frequency rows are uncorrelated and
$\mathbb{E}[\sin^2]=\mathbb{E}[\cos^2]=\tfrac12$ give
$\mathbb{E}[G] = \tfrac12 I$, but only in expectation and only at a well-chosen
bandwidth: as $\sigma_B\to0$ the features collapse toward rank one and
$\kappa(G)$ diverges, while at large $\sigma_B$ they decorrelate faster than the
target varies and the model must fit high-frequency noise. Bandwidth tuning in
FEDONet is, in this language, conditioning repair.

The Chebyshev dictionary needs no such repair, and the reason is exact rather
than statistical. Under the weight $w(\xi)=(1-\xi^2)^{-1/2}$ that makes
$\{T_n\}$ orthogonal, $G$ is diagonal with a condition number independent of
$d_{\mathrm{trunk}}$, of $K$ and of any random draw. The query set is sampled
uniformly rather than under that weight, so the realized conditioning is weaker
than the orthogonal figure, and this is a property of the sampling measure that
grid refinement does not remove (\Cref{app:cheb}). What carries the comparison
is therefore not the magnitude of $\kappa(G)$ but its character. The polynomial
dictionary is not uniformly better conditioned, at a well-chosen bandwidth
the Fourier features are in fact better, and at a poor one they are singular to
single precision, the realized $\kappa(G)$ spanning many orders of magnitude
across plausible $\sigma_B$ (\Cref{app:rff}). It is instead unconditionally conditioned: a fixed number, known before training,
identical across seeds, with no hyperparameter to select. The whitening that
\cite{sojitra2025} obtains in expectation and at a tuned bandwidth, the
polynomial dictionary obtains deterministically.

Alongside conditioning, the three embeddings differ in what they represent
cheaply. A Chebyshev dictionary represents a linear trend exactly in two modes;
a random Fourier dictionary must build the same trend from periodic functions
and pays the Gibbs penalty at the interval endpoints; a raw trunk must build
it, like everything else, from nonlinearities. On a problem whose leading
behaviour in $\xi_y$ and $\xi_t$ is a monotone drift, the bubble rising, the
plume tip advancing, the entrained volume growing, this difference is not
incidental, and the per-channel evidence of \Cref{sec:results} identifies it as
the operative mechanism. \Cref{tab:embed-compare} summarizes the comparison.

\begin{table}[ht!]
\centering
\small
\setlength{\tabcolsep}{4pt}
\renewcommand{\arraystretch}{1.3}
\caption{The three trunk embeddings side by side. Everything not listed
(branch, trunk depth and width, latent width $p$, loss, optimizer, schedule
and data split) is identical across the three models.}
\label{tab:embed-compare}
\begin{tabularx}{\textwidth}{@{}l L{0.60} L{1.15} L{1.25}@{}}
\toprule
 & \textbf{DeepONet} & \textbf{FEDONet} & \textbf{SEDONet} \\
\midrule
Embedding $\phi$
  & $\bs{\xi}$
  & $[\sin,\cos](2\pi B\bs{\xi})$
  & $\{T_k T_l T_m\}_{\mathcal{A}}$ \\
Implied prior
  & none (spectral bias)
  & periodic, translation invariant
  & non-periodic, boundary emphasized \\
Represents a linear trend
  & via nonlinearities
  & slowly converging series
  & exactly, in two modes \\
\bottomrule
\end{tabularx}
\end{table}
\subsection{Architectures and parameter accounting}
\label{sec:opnet:arch}

\Cref{fig:arch} presents the three architectures in one layout, drawn so that
the single point of difference, the block on the trunk input path, is the
only thing that changes between them. The branch is
$9\to512\to512\to5p$ and the trunk $d_{\mathrm{trunk}}\to512\to512\to p$ with
$p=64$ latent basis channels and GELU activations \cite{hendrycks2016gelu},
identical across the three models except for $d_{\mathrm{trunk}}$.

\begin{figure}[H]
\centering
\includegraphics[width=\textwidth]{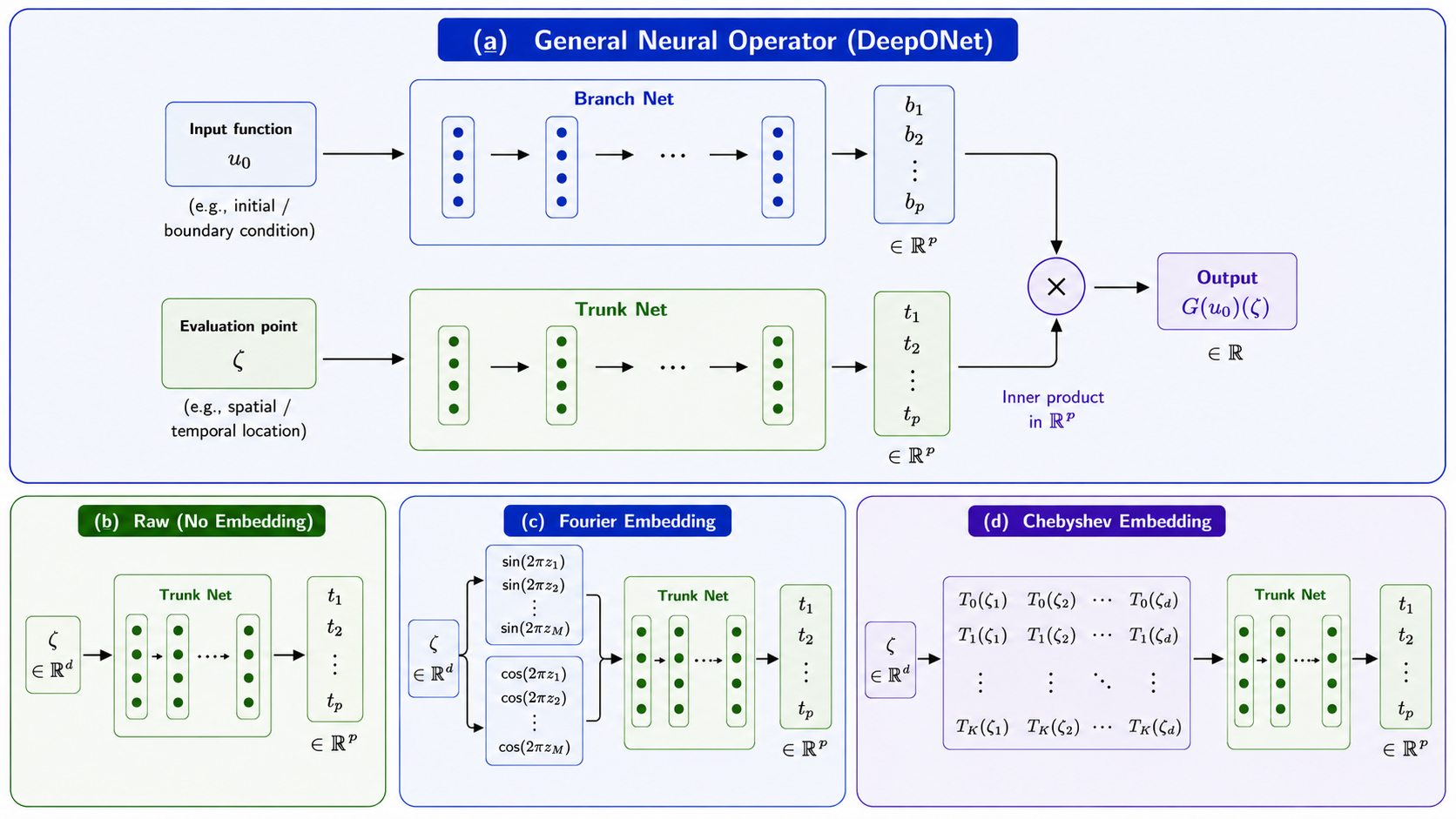}
\caption{The three architectures differ only in the trunk coordinate
embedding. \textbf{(a)} The branch-trunk factorization of \cref{eq:donet}:
the branch encodes the input, the trunk encodes the query coordinate after a
fixed embedding $\phi$, and the two are combined by an inner product.
\textbf{(b)-(d)} The three fillings of the $\phi$ slot, raw rescaled
coordinates, random Fourier features, and a deterministic Chebyshev dictionary, which constitute the entire experimental manipulation. The first trunk
layer has the same width in both (c) and (d), so FEDONet and SEDONet are
parameter-identical, and no block carries trainable parameters.}
\label{fig:arch}
\end{figure}

Since the embeddings carry no trainable parameters, the three models differ
only in the input width of that first trunk layer, giving $0.73$~M parameters
for the raw baseline and $0.86$~M for the two embedded models
(\Cref{tab:arch}). The gap is precisely one weight matrix, about $17.8\%$ of
the baseline. That residual advantage is unavoidable if the embedding width is
to be nontrivial and should be acknowledged rather than explained away, but it
cannot account for the result: FEDONet and SEDONet have identical parameter
counts, depths, widths, optimizers and data, so whatever separates them is the
content of $\phi$ and nothing else. The DeepONet comparison is the weaker of
the two and we report it as such.

\begin{table}[ht!]
\centering
\caption{Model configurations. The three architectures share branch and trunk
depths and widths and differ only in the trunk coordinate embedding; the
parameter gap is exactly the first trunk layer.}
\label{tab:arch}
\begin{tabular}{llcc}
\toprule
Model    & Trunk embedding $\phi$                                    & $d_{\mathrm{trunk}}$ & Parameters \\
\midrule
DeepONet & raw $\bs{\xi}=(\xi_x,\xi_y,\xi_t)$                        & $3$   & $0.73$~M \\
FEDONet  & $[\sin(2\pi B\bs{\xi}),\cos(2\pi B\bs{\xi})]$, $M=128$    & $256$ & $0.86$~M \\
SEDONet  & $T_k(\xi_x)T_l(\xi_y)T_m(\xi_t)$, graded, $K=10$          & $256$ & $0.86$~M \\
\bottomrule
\end{tabular}
\end{table}

Neither embedding changes the asymptotic cost of training. The embedding adds a
wider first trunk layer, a few percent of the per-step cost, and for SEDONet
the dictionary itself is evaluated once for the whole training run or
regenerated on the fly from the recurrence \cref{eq:cheb-rec} at negligible
cost.

\subsection{Training and evaluation protocol}
\label{sec:opnet:protocol}

The ensemble is split at the level of the design member into training,
validation and test sets of $905$, $16$ and $103$ runs, so no test parameter
vector is seen during training. Each of the five channels is standardized to
zero mean and unit variance using training-split statistics, and the closure is
applied in physical units per \cref{eq:closure-std}. The objective is the
channel-averaged mean-squared error over a mini-batch $\mathcal{B}$ of ensemble
members and the query set,
\begin{equation}
\mathcal{L}(\theta)
= \frac{1}{|\mathcal{B}|\,Q}\sum_{i\in\mathcal{B}}\sum_{q=1}^{Q}
  \frac{1}{5}\sum_{c=1}^{5}
  \Big(\widehat{s}_c\big(\vzeta^{q};\vp^{(i)}\big) - s^{(i)}_c(\vzeta^{q})\Big)^{2},
\label{eq:loss}
\end{equation}
evaluated on standardized channels, with the third concentration term acting
through the metric \cref{eq:closure-metric}. Optimization uses Adam
\cite{kingma2015adam} with cosine decay of the learning rate to zero. All three
models are trained under an identical schedule, validation relative $L^2$ is
recorded periodically on the held-out runs, and every reported test number
comes from the final weights rather than a selected checkpoint, so no
comparison depends on where one might have stopped.

The primary metric is the relative $L^2$ error over the full space-time grid,
computed per run and per channel and then averaged over the test runs,
\begin{equation}
\relL^{(c)}
= \frac{\big\|\widehat{s}_c - s_c\big\|_{2}}{\big\|s_c\big\|_{2}},
\qquad
\relL = \frac{1}{5}\sum_{c=1}^{5}\relL^{(c)} ,
\label{eq:relL2}
\end{equation}
with norms taken in physical rather than standardized units so that the numbers
are comparable across channels of different magnitude. This metric is
scale-sensitive by construction: a channel whose norm is small, such as $c_a$,
is judged more harshly for the same absolute error, and the denominator
vanishes identically for the velocity channels at $t=0$, where the initial
condition is at rest, a point we return to in \Cref{sec:res:sample}.
Alongside it we report the simplex violation
\begin{equation}
\mathcal{E}_{\Delta} = \max_{\vx,t}\Big|\sum_i \widehat{c}_i - 1\Big| ,
\label{eq:simplexviol}
\end{equation}
the wall-normal error profile, averaged over $x$, over the snapshot set
$\mathcal{T}$ and over the test runs,
\begin{equation}
e(y) = \frac{1}{|\mathcal{T}|\,N_x}\sum_{t\in\mathcal{T}}\sum_{j}
       \Big|\widehat{c}(x_j,y,t) - c(x_j,y,t)\Big| ,
\label{eq:profile}
\end{equation}
the domain-averaged phase error at each output time, which is
\cref{eq:profile} with the roles of $y$ and $t$ exchanged, and the deviations
of the derived physical diagnostics recomputed from the predicted fields.
\subsection{Algorithms}
\label{sec:opnet:algorithms}
\Cref{alg:train-deeponet,alg:train-fedonet,alg:train-sedonet} give one training
procedure per method. They are written in parallel deliberately: every line is
shared except the embedding block at the top, which is the whole experimental
manipulation. What differs is when the coordinate representation is fixed.
DeepONet has none to fix. FEDONet fixes a random draw at initialization, so its
representation is frozen for the run but differs between runs and depends on a
chosen bandwidth. SEDONet fixes a deterministic dictionary constructed before
training begins, bit-identical across runs and specified by two integers.

\begin{algorithm}[ht!]
\caption{DeepONet --- Raw-coordinate trunk}
\label{alg:train-deeponet}
\begin{algorithmic}[1]
\Require Training ensemble $\mathcal{D}=\{(\vp^{(i)},\bs{s}^{(i)})\}_{i=1}^{N_{\mathrm{tr}}}$;
query set $Z=\{\vzeta^q\}_{q=1}^{Q}$; latent width $p$; learning-rate schedule
$\{\eta_e\}$; epochs $E$
\Ensure Trained $\theta=(\theta_{\mathrm{branch}},\theta_{\mathrm{trunk}})$
\State $\Phi \gets \bs{\xi}(Z) \in \R^{Q\times d_{\mathrm{trunk}}}$
       \Comment{embedding: affine rescaling only, \cref{eq:affine,eq:phi-raw}}
\State Standardize channels using training-split statistics; initialize $\theta$
\For{$e = 1,\dots,E$}
  \For{each mini-batch $\mathcal{M}\subset\mathcal{D}$}
    \State $\bs{B}(\vp^{(i)}) \gets \Bnet(\vp^{(i)}) \in \R^{5\times p}$
           for $i\in\mathcal{M}$ \Comment{branch}
    \State $\bs{T} \gets \Tnet(\Phi) \in \R^{Q\times p}$
           \Comment{trunk; first layer $d_{\mathrm{trunk}}\to512$}
    \State $\widehat{\bs{S}}^{(i)} \gets
           \bs{T}\,\bs{B}(\vp^{(i)})^{\!\top} + \bs{1}\bs{b}_0^{\!\top}$;
           apply the closure \cref{eq:closure-std}
    \State Evaluate $\mathcal{L}(\theta)$ by \cref{eq:loss}; take one Adam step
           with rate $\eta_e$
  \EndFor
  \If{$e \bmod 5 = 0$}
    \State Evaluate validation $\relL$ on the held-out runs
  \EndIf
\EndFor
\end{algorithmic}
\end{algorithm}

\begin{algorithm}[ht!]
\caption{FEDONet --- Random Fourier feature trunk}
\label{alg:train-fedonet}
\begin{algorithmic}[1]
\Require As \Cref{alg:train-deeponet}, plus feature count $M$, bandwidth
$\sigma_B$, embedding seed
\Ensure Trained $\theta$; frozen $B$
\State Draw $B\in\R^{M\times3}$ with $B_{ij}\sim\mathcal{N}(0,\sigma_B^2)$ and
       freeze it \Comment{drawn once; $B\notin\theta$}
\State $\Phi \gets \big[\sin(2\pi\,\bs{\xi}(Z)B^{\!\top}),\;
       \cos(2\pi\,\bs{\xi}(Z)B^{\!\top})\big] \in \R^{Q\times d_{\mathrm{trunk}}}$
       \Comment{embedding, \cref{eq:phi-fourier}}
\State Standardize channels using training-split statistics; initialize $\theta$
\For{$e = 1,\dots,E$}
  \For{each mini-batch $\mathcal{M}\subset\mathcal{D}$}
    \State $\bs{B}(\vp^{(i)}) \gets \Bnet(\vp^{(i)})$ for $i\in\mathcal{M}$
           \Comment{branch}
    \State $\bs{T} \gets \Tnet(\Phi)$
           \Comment{trunk; first layer $d_{\mathrm{trunk}}\to512$}
    \State $\widehat{\bs{S}}^{(i)} \gets
           \bs{T}\,\bs{B}(\vp^{(i)})^{\!\top} + \bs{1}\bs{b}_0^{\!\top}$;
           apply the closure \cref{eq:closure-std}
    \State Evaluate $\mathcal{L}(\theta)$ by \cref{eq:loss}; take one Adam step
           with rate $\eta_e$
  \EndFor
  \If{$e \bmod 5 = 0$}
    \State Evaluate validation $\relL$ on the held-out runs
  \EndIf
\EndFor
\end{algorithmic}
\end{algorithm}

\begin{algorithm}[ht!]
\caption{SEDONet --- Chebyshev spectral trunk}
\label{alg:train-sedonet}
\begin{algorithmic}[1]
\Require As \Cref{alg:train-deeponet}, plus per-axis degree bound $K$ and
feature budget $d_{\mathrm{trunk}}$
\Ensure Trained $\theta$
\State $\mathcal{A} \gets$ the first $d_{\mathrm{trunk}}$ multi-indices
       $(k,l,m)$ with $\max(k,l,m)<K$ in order of increasing total degree
       $k+l+m$, ties broken lexicographically \Comment{\cref{eq:graded}}
\State $\Phi[q,j] \gets T_{k_j}(\xi_x^q)\,T_{l_j}(\xi_y^q)\,T_{m_j}(\xi_t^q)$
       for $q=1,\dots,Q$ and $(k_j,l_j,m_j)\in\mathcal{A}$, each factor from
       the recurrence \cref{eq:cheb-rec}
       \Comment{embedding; no parameters, no random draw}
\State Standardize channels using training-split statistics; initialize $\theta$
\For{$e = 1,\dots,E$}
  \For{each mini-batch $\mathcal{M}\subset\mathcal{D}$}
    \State $\bs{B}(\vp^{(i)}) \gets \Bnet(\vp^{(i)})$ for $i\in\mathcal{M}$
           \Comment{branch}
    \State $\bs{T} \gets \Tnet(\Phi)$
           \Comment{trunk; first layer $d_{\mathrm{trunk}}\to512$}
    \State $\widehat{\bs{S}}^{(i)} \gets
           \bs{T}\,\bs{B}(\vp^{(i)})^{\!\top} + \bs{1}\bs{b}_0^{\!\top}$;
           apply the closure \cref{eq:closure-std}
    \State Evaluate $\mathcal{L}(\theta)$ by \cref{eq:loss}; take one Adam step
           with rate $\eta_e$
  \EndFor
  \If{$e \bmod 5 = 0$}
    \State Evaluate validation $\relL$ on the held-out runs
  \EndIf
\EndFor
\end{algorithmic}
\end{algorithm}

The embedding block of \Cref{alg:train-sedonet} is executed once for the whole
training run and is identical across seeds, so the trunk sees a spectral input
distribution whose conditioning is fixed before training begins rather than
realized from a draw (\Cref{app:cheb:uniform}). This is the one property
\Cref{alg:train-fedonet} cannot share, since its first line depends on a random
draw whose bandwidth must additionally be selected.
\section{Results}
\label{sec:results}

The evidence is reported in increasing order of specificity: what accuracy the
three models reach and how it is distributed over channels and over the
parameter space (\Cref{sec:res:accuracy}); where in space and time the
differences between them live (\Cref{sec:res:localization}); one hard member
examined snapshot by snapshot (\Cref{sec:res:sample}); and whether the
pointwise gains translate into derived physical quantities and constraint
consistency (\Cref{sec:res:physical}). The interpretation is developed
alongside the evidence, because each piece of it discriminates between
different candidate explanations.

\subsection{Test accuracy, by channel and across the design}
\label{sec:res:accuracy}
SEDONet reduces the mean relative $L^2$ error by $24.0\%$ relative to the
raw-trunk baseline and by $16.8\%$ relative to FEDONet, uniformly across
channels. That uniformity is not automatic on a multi-field target: a prior
that helped the phase fields at the expense of the velocities would be a more
equivocal outcome, and the shared trunk of \Cref{sec:opnet:factorization} is
exactly the arrangement that would expose such a trade-off.

The per-channel reductions of \Cref{tab:reductions} are the most informative
result in the paper, and they are invisible in the channel-averaged number. The
Fourier embedding delivers its gain almost entirely on $c_a$ and the velocities, the channels with small-scale, oscillatory, or recirculating structure,
which is exactly where a stationary broadband kernel should help and close
to none of it on the two large-region concentration fields. The Chebyshev
embedding improves those same oscillatory channels by more, and additionally improves $c_w$ and $c_o$ by about a fifth; those two are
precisely the fields dominated by a large connected region bounded by a single
slowly translating level set, that is, by a monotone and strongly non-periodic
trend in $y$ and in $t$.

That split separates the two candidate mechanisms cleanly, and it is the
central interpretive claim of the paper. Broadband resolution, supplying
high-wavenumber modes at initialization instead of requiring the network to
synthesize them, is common to both embeddings and accounts for the
improvement both deliver on $c_a$ and on the velocities. Trend representation
is specific to the polynomial dictionary: a linear ramp is two Chebyshev modes,
whereas a dictionary built from functions periodic on the embedding scale
represents the same ramp as a slowly converging sinusoidal series, spending on
the trend part of a budget the polynomial basis leaves free for the interface
structure.

\begin{table}[ht!]
\centering
\caption{Test relative $L^2$ error over $103$ held-out simulations, and the
maximum simplex violation $\mathcal{E}_{\Delta}$ of \cref{eq:simplexviol}.
Lowest per column in bold.}
\label{tab:test}
\begin{tabular}{lcccccc c}
\toprule
Model    & mean & $c_w$ & $c_o$ & $c_a$ & $u$ & $v$ & $\mathcal{E}_{\Delta}$ \\
\midrule
DeepONet & 0.0821 & 0.0319 & 0.0170 & 0.1331 & 0.1406 & 0.0879 & $1.79\times10^{-7}$ \\
FEDONet  & 0.0750 & 0.0315 & 0.0170 & 0.1236 & 0.1231 & 0.0797 & $1.79\times10^{-7}$ \\
SEDONet  & \textbf{0.0624} & \textbf{0.0251} & \textbf{0.0136} & \textbf{0.1053} & \textbf{0.1037} & \textbf{0.0644} & $1.79\times10^{-7}$ \\
\midrule
\multicolumn{8}{l}{SEDONet vs.\ DeepONet: $-24.0\%$ \quad
                   SEDONet vs.\ FEDONet: $-16.8\%$ \quad
                   channels improved: $5/5$} \\
\bottomrule
\end{tabular}
\end{table}

\begin{table}[ht!]
\centering
\caption{Per-channel error reduction (\%) implied by \Cref{tab:test}. The two
embeddings differ in where their gains fall as well as in magnitude: the
Fourier trunk concentrates its improvement on the oscillatory channels, while
the Chebyshev trunk improves every channel by a comparable fifth to a quarter.}
\label{tab:reductions}
\begin{tabular}{lcccccc}
\toprule
Comparison & mean & $c_w$ & $c_o$ & $c_a$ & $u$ & $v$ \\
\midrule
FEDONet vs.\ DeepONet & $8.6$  & $1.3$  & $0.0$  & $7.1$  & $12.4$ & $9.3$ \\
SEDONet vs.\ DeepONet & $24.0$ & $21.3$ & $20.0$ & $20.9$ & $26.2$ & $26.7$ \\
SEDONet vs.\ FEDONet  & $16.8$ & $20.3$ & $20.0$ & $14.8$ & $15.8$ & $19.2$ \\
\bottomrule
\end{tabular}
\end{table}

\begin{figure}[H]
\centering
\includegraphics[width=0.42\textwidth]{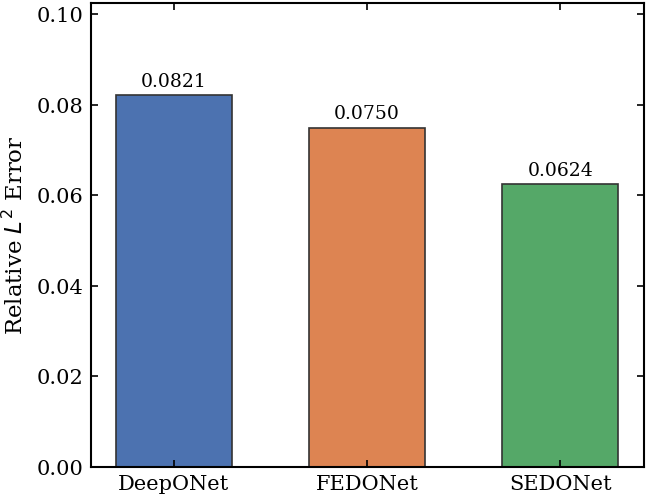}
\hfill
\includegraphics[width=0.55\textwidth]{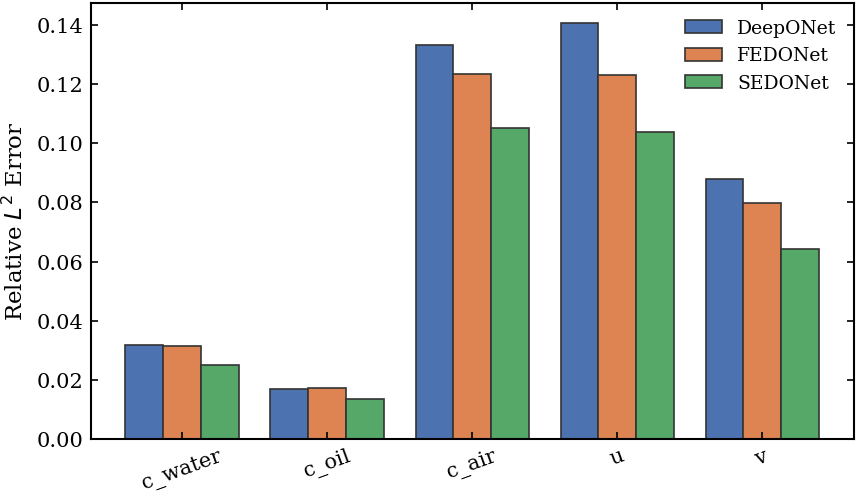}
\caption{Test relative $L^2$ over the $103$ held-out simulations.
\textbf{(a)} Channel-averaged error. \textbf{(b)} Per-channel breakdown. The
ordering SEDONet $<$ FEDONet $<$ DeepONet holds for the mean and for every one
of the five channels, but the margins are strongly channel dependent:
FEDONet is close to the baseline on $c_w$ and $c_o$, whereas SEDONet improves
those two channels by about a fifth.}
\label{fig:accuracy}
\end{figure}

The absolute level of the error is also interpretable. The oil field $c_o$ is
easiest, one large connected region with a single smooth boundary, while $c_a$
and the horizontal velocity $u$ carry the budget: $c_a$ is a small, strongly
deforming, topologically active object whose support is a few percent of the
domain, and $u$ is driven by the lateral squeezing of the plume and the
recirculation behind the cap. Notably the velocity channels are where SEDONet
gains most, at $26$--$27\%$, which is not what one would predict if the
mechanism were purely sharp-interface resolution in the phase fields.

Two cautions attach to these numbers. The relative $L^2$ metric is unforgiving
for $c_a$, because the small norm of a field occupying a few percent of the
domain puts a small denominator under any error concentrated on the bubble rim;
that column is a stringent interface-localization score rather than evidence
that the bubble is misplaced, which \Cref{fig:fields} shows it is not. And the
common accuracy floor in the range $6$--$8\%$ is most likely not a shortcoming
of any of the three: a rank-$p$ separated representation is bounded below by
the Kolmogorov $p$-width of the solution manifold
(\Cref{sec:opnet:factorization}), and that width decays slowly for
advection-dominated problems with moving interfaces.

\begin{figure}[H]
\centering
\includegraphics[width=0.98\textwidth]{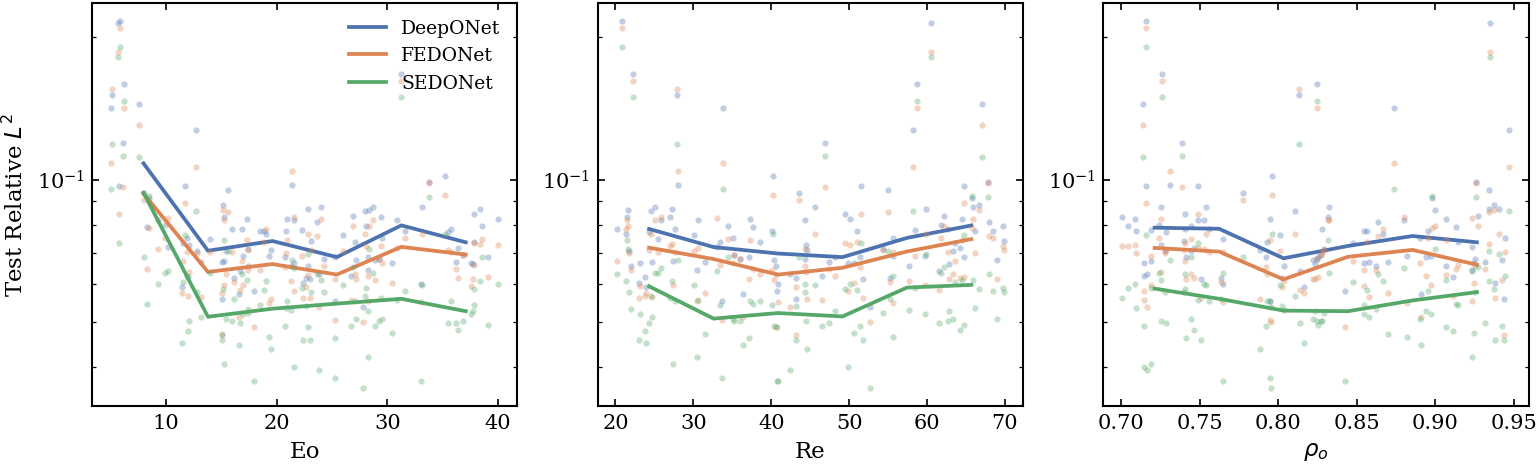}
\caption{Test relative $L^2$ of each held-out member (points) against the three
most influential design parameters, with binned means (lines), log scale. The
SEDONet curve lies below the other two across the entire range of all three
parameters and the vertical separation is roughly constant in the log, so the
improvement is multiplicative and essentially uniform in $\vp$ rather than
concentrated in a subregion of the design.}
\label{fig:errvsparam}
\end{figure}

\Cref{fig:errvsparam} answers a question the aggregate numbers cannot: is the
$24\%$ gain broad, or an average over a strong improvement on part of the
design and none elsewhere? It is the former. Across the E\"otv\"os number, the
Reynolds number and the density ratio the three binned-mean curves are close to
parallel on a logarithmic axis with SEDONet uniformly lowest, so the parameter
dependence is a property of the problem rather than of any one
architecture. That shared dependence is itself informative: error rises sharply
as $\mathrm{Eo}\to5$, where strong water-air tension produces a taller,
thinner plume and a more sharply curved bubble crown and therefore more
high-wavenumber content in every channel, and the ordering mirrors the
sensitivity ranking of the reference data itself. The scatter within each bin
spans nearly an order of magnitude, which is why the single-member comparison
of \Cref{sec:res:sample} must be read as an illustration rather than as
evidence.

\subsection{Where the error lives, and when}
\label{sec:res:localization}
A scalar error can conceal entirely different outcomes, and for interfacial
physics only one of them is worth having: a $24\%$ reduction is compatible with
a uniform improvement everywhere, with a large improvement in a small region,
or with an improvement in the bulk purchased at the cost of the interfaces.
\Cref{fig:errorloc} resolves which.

In space, the wall-normal profile is a two-peak structure that tracks the
moving geometry, with peaks at the bubble rim and at the interface that are
narrow and nearly coincident across the three models until $t=2$, after which
the upper peak detaches and migrates with the bubble while a second, broader
peak persists at the deformed interface and the plume flanks. It is at these
peaks, and only at these peaks, that the models separate: the baseline reaches
$0.33$ at $t=4$ against SEDONet's $0.17$, while in the bulk of both liquid
layers all three curves lie together near zero. The improvement is therefore
not a uniform rescaling of a smooth error field but a targeted improvement in
interface resolution.

The same panels falsify the explanation one would reach for first. The error is
small at the walls, at every time and for every model. This matters
because the classical justification for a Chebyshev basis is boundary-layer
resolution, and because the polynomial dictionary does carry its largest prior
amplitude at the domain boundary, the mechanism is available to the model.
It is demonstrably not what produces the gain, because there is no wall error
to remove: the viscous layers that form at the no-slip walls are three to five
times wider than the diffuse interface, carry no high-wavenumber content, and
the bubble never approaches them. Boundary emphasis is a real property of the
polynomial dictionary and is simply not what this problem rewards.

In time, the lower panels show the domain-averaged phase error against $t$. For
$t\lesssim1.5$, while the bubble is simply rising through the water, the three
curves are close and the raw-coordinate baseline is in fact marginally the best of the three: the target in that window is a single, smoothly
translating, nearly axisymmetric object that a raw trunk represents adequately,
and the embedded models pay a small price for carrying a dictionary they do not
yet need. The curves cross at breakthrough, and from there the ordering inverts
and the gap widens monotonically to the end of the horizon.

\begin{figure}[H]
\centering
\includegraphics[width=0.98\textwidth]{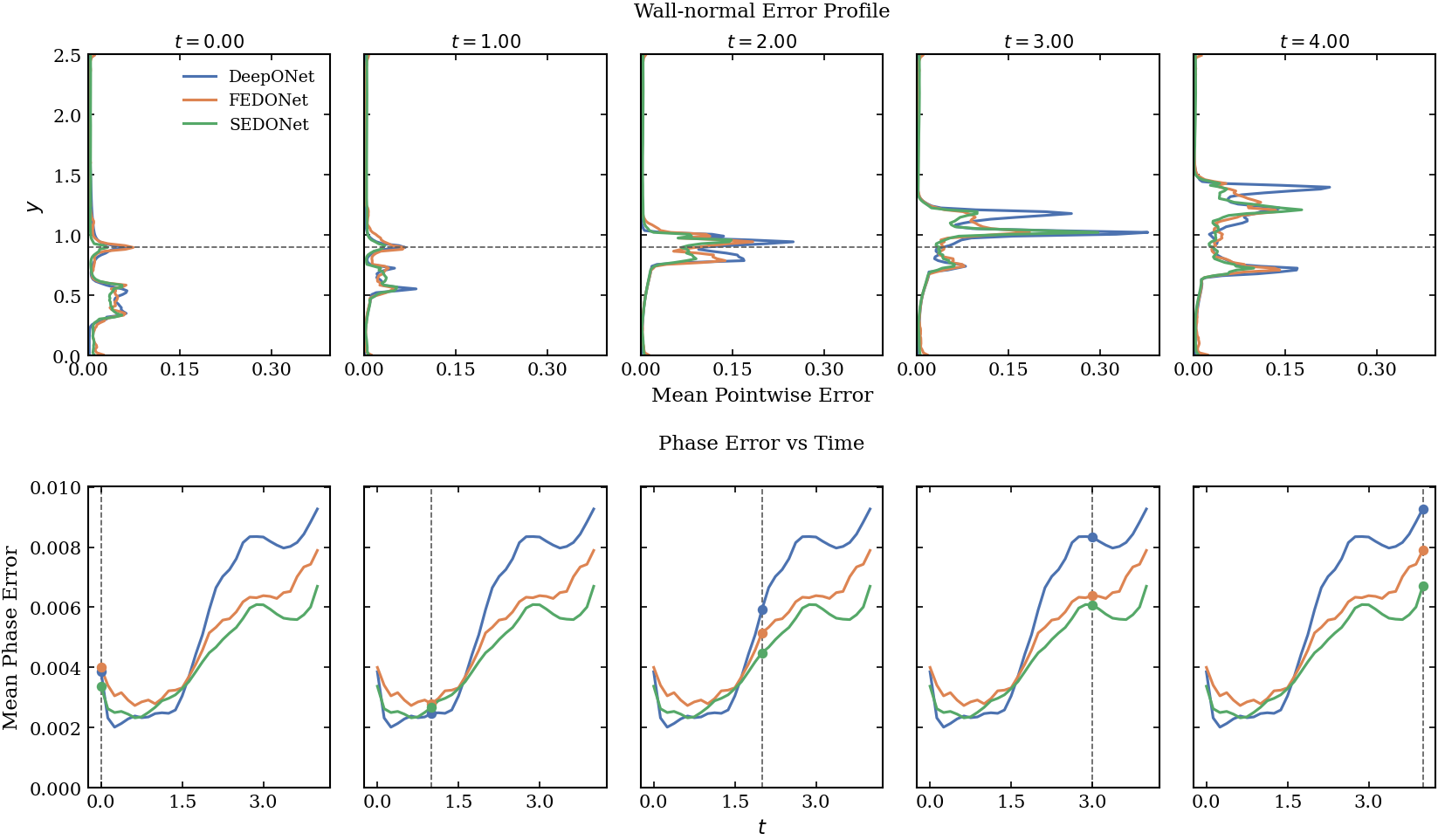}
\caption{Top: wall-normal error profile $e(y)$ from \cref{eq:profile} at
$t=0,1,2,3,4$; the dashed line marks the undisturbed water--oil interface
$y_{\mathrm{int}}=0.9$. Bottom: domain-averaged phase error against time, with
the vertical dashed line marking the snapshot shown above it. Error is
concentrated at the interfaces, near zero in both bulk layers and at both
walls, and the three models separate only after breakthrough.}
\label{fig:errorloc}
\end{figure}

The timing of that crossing is a direct test of the mechanism rather than an
incidental observation. Before breakthrough the target is smooth and
effectively low-rank, cheap to represent in any reasonable basis, and the trunk
is not the binding constraint on accuracy. After breakthrough it acquires a
triple junction, a thin high-curvature spike and a topological transition as
the bubble detaches, and simultaneously the monotone drift in the wall-normal
and temporal coordinates becomes pronounced. These are precisely the conditions
under which both advantages identified in \Cref{sec:res:accuracy} should come
into play, and this is precisely when the curves separate. The trunk embedding
does not make the surrogate uniformly better; it makes it better at the
topologically active phase of the flow, which is the phase where all three
surrogates are least accurate and the phase anyone runs this simulation to see.

\subsection{A hard test member, resolved in time}
\label{sec:res:sample}

Test member $269$ ranks in the hardest decile of the test set, for reasons the
previous subsections make legible. Two of its nine parameters sit at the edge
of the design box: the E\"otv\"os number is at the extreme low end of its
range, the corner where \Cref{fig:errvsparam} shows all three models degrading;
and the normalized oil spreading coefficient is near its lower bound, so the
configuration lies close to the total-wetting boundary at which oil would film
between water and air, leaving the triple junction nearly degenerate.

The $t=0$ row of \Cref{tab:sample269-time} needs comment before the rest can be
read, because one of its two anomalies is an artifact and the other is not. The
per-snapshot relative $L^2$ of the velocity channels at $t=0$ returns enormous
values in all three models, but these are not failures: the initial condition
sets the velocity identically to zero, so the denominator of \cref{eq:relL2}
vanishes and the figures carry no information about accuracy. We therefore
exclude $t=0$ from all per-snapshot velocity metrics; the space-time aggregate
of \Cref{tab:sample269} is unaffected, since there the denominator is dominated
by the finite velocities at later times. We flag this because it is a trap in
any operator-learning benchmark initialized from rest.

\begin{figure}[H]
\centering
\includegraphics[width=0.80\textwidth]{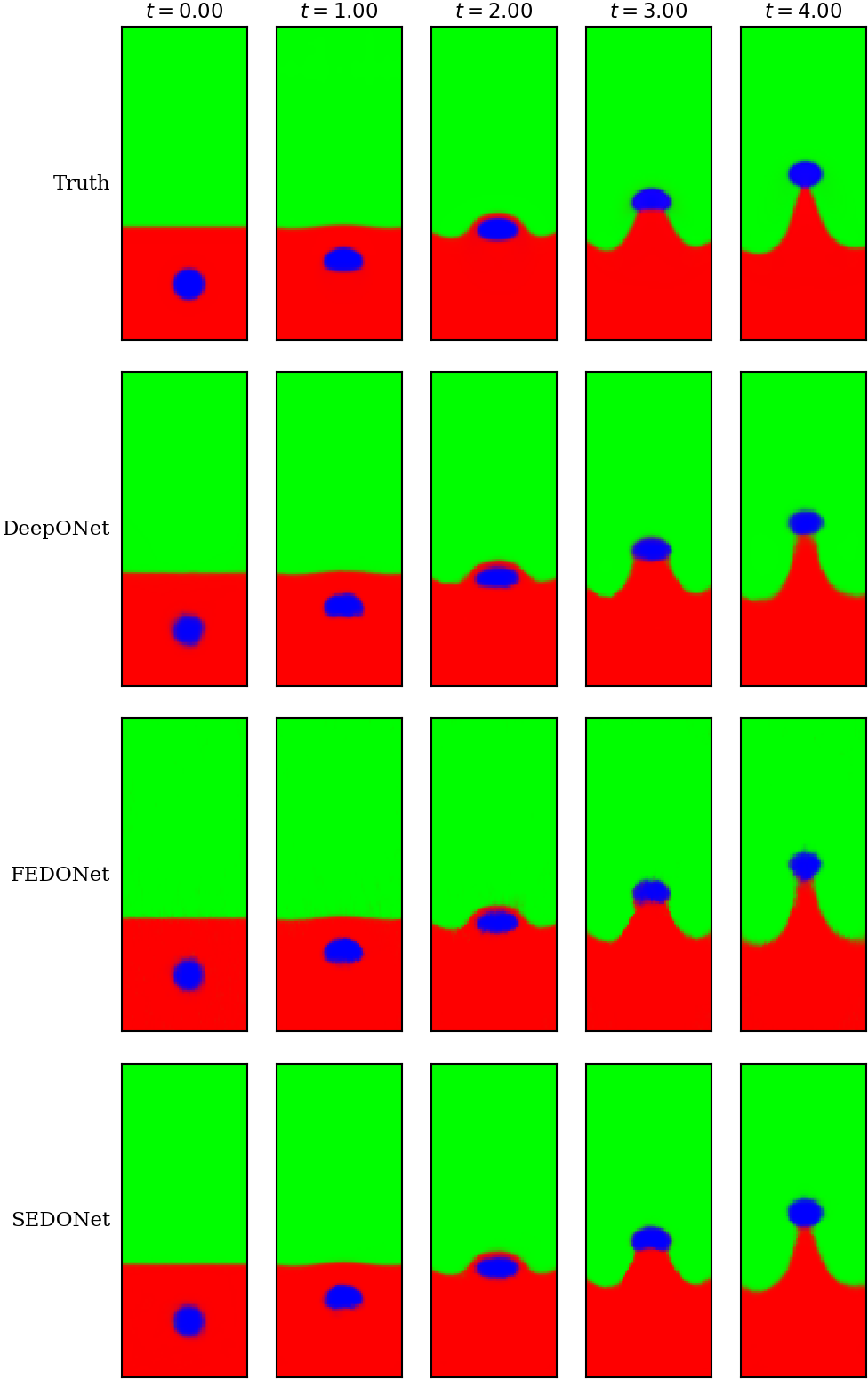}
\caption{Test member $269$: RGB composite of $(c_w,c_o,c_a)$ (red $=$ water,
green $=$ oil, blue $=$ air) at $t=0,1,2,3,4$. Rows: reference, DeepONet,
FEDONet, SEDONet. All three surrogates reproduce the full sequence, rise,
cap deformation, breakthrough, plume entrainment, detachment, and the
correct topology at every time, including the reconnected oil layer and the
sharp water spike beneath the detached bubble at $t=4$.}
\label{fig:fields}
\end{figure}

\begin{figure}[H]
\centering
\includegraphics[width=0.90\textwidth]{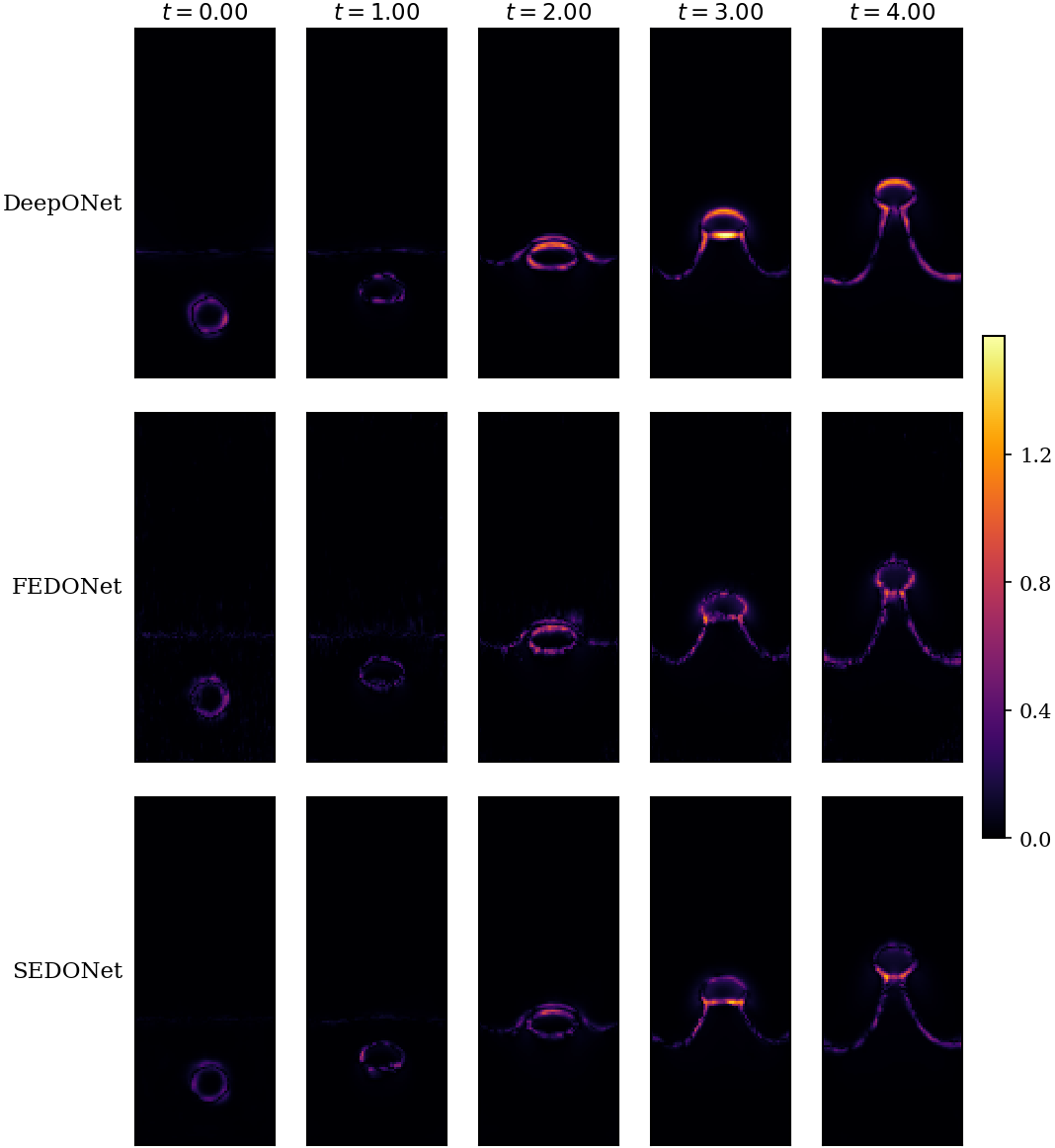}
\caption{Absolute phase error for the same member on a shared colour scale.
The error is a thin shell hugging the bubble rim, the water-oil interface and
the plume flanks, and is black over the bulk of both liquid layers. The
baseline's rim is a thick bright ring; SEDONet's is a thinner, dimmer arc. The
one persistent bright feature common to all three is the horizontal band
immediately beneath the bubble crown at $t\ge2$, where the drainage film
between the bubble and the oil is thinnest.}
\label{fig:fields-error}
\end{figure}

\begin{table}[ht!]
\centering
\caption{Relative $L^2$ error on test member $269$, resolved by snapshot. The
$t=0$ velocity entries are undefined, for the reason given in the text, and are
marked $\dagger$. Lowest per row in bold.}
\label{tab:sample269-time}
\begin{tabular}{llcccccc}
\toprule
$t$ & Model & mean & $c_w$ & $c_o$ & $c_a$ & $u$ & $v$ \\
\midrule
\multirow{3}{*}{$0$}
 & DeepONet & $\dagger$ & 0.0337 & 0.0073 & 0.1439 & $\dagger$ & $\dagger$ \\
 & FEDONet  & $\dagger$ & 0.0320 & 0.0093 & 0.1275 & $\dagger$ & $\dagger$ \\
 & SEDONet  & $\dagger$ & \textbf{0.0236} & \textbf{0.0045} & \textbf{0.1056} & $\dagger$ & $\dagger$ \\
\midrule
\multirow{3}{*}{$1$}
 & DeepONet & 0.0612 & 0.0244 & 0.0067 & 0.1004 & \textbf{0.1067} & 0.0679 \\
 & FEDONet  & 0.0618 & \textbf{0.0231} & 0.0089 & \textbf{0.0865} & 0.1148 & 0.0755 \\
 & SEDONet  & \textbf{0.0609} & 0.0249 & \textbf{0.0046} & 0.1069 & 0.1075 & \textbf{0.0605} \\
\midrule
\multirow{3}{*}{$2$}
 & DeepONet & 0.1563 & 0.0718 & 0.0265 & 0.2874 & 0.2681 & 0.1276 \\
 & FEDONet  & 0.1221 & 0.0534 & 0.0206 & 0.2121 & 0.2121 & 0.1121 \\
 & SEDONet  & \textbf{0.0985} & \textbf{0.0408} & \textbf{0.0167} & \textbf{0.1597} & \textbf{0.1865} & \textbf{0.0890} \\
\midrule
\multirow{3}{*}{$3$}
 & DeepONet & 0.2160 & 0.0752 & 0.0515 & 0.3803 & 0.3675 & 0.2056 \\
 & FEDONet  & \textbf{0.1399} & \textbf{0.0479} & 0.0366 & \textbf{0.1730} & 0.2771 & 0.1650 \\
 & SEDONet  & 0.1446 & 0.0494 & \textbf{0.0335} & 0.2147 & \textbf{0.2689} & \textbf{0.1564} \\
\midrule
\multirow{3}{*}{$4$}
 & DeepONet & 0.1685 & 0.0601 & 0.0641 & 0.2927 & 0.2174 & 0.2081 \\
 & FEDONet  & 0.1333 & 0.0507 & 0.0452 & \textbf{0.1831} & 0.2406 & 0.1472 \\
 & SEDONet  & \textbf{0.1192} & \textbf{0.0398} & \textbf{0.0412} & 0.1986 & \textbf{0.1800} & \textbf{0.1364} \\
\bottomrule
\end{tabular}
\end{table}

\begin{table}[ht!]
\centering
\caption{Relative $L^2$ error on member $269$ aggregated over the full
space--time grid, that is, with a single global denominator rather than a
per-snapshot one. Errors are roughly $50\%$ above the test mean of
\Cref{tab:test} and the overall ordering is restored.}
\label{tab:sample269}
\begin{tabular}{lcccccc}
\toprule
Model    & mean & $c_w$ & $c_o$ & $c_a$ & $u$ & $v$ \\
\midrule
DeepONet & 0.1421 & 0.0546 & 0.0359 & 0.2527 & 0.2233 & 0.1440 \\
FEDONet  & 0.1083 & 0.0427 & 0.0259 & 0.1645 & 0.1892 & 0.1192 \\
SEDONet  & \textbf{0.0954} & \textbf{0.0350} & \textbf{0.0230} & \textbf{0.1555} & \textbf{0.1624} & \textbf{0.1011} \\
\bottomrule
\end{tabular}
\end{table}

The concentration errors at $t=0$ are a different matter, and more interesting.
They are not small. Yet the initial phase field is an analytic function of the
parameters, so a surrogate handed $\vp$ knows the initial state exactly
and is nonetheless placing the bubble rim a fraction of a cell off. This is
pure representation error in the trunk, present before any dynamics have been
learned, and it is the cleanest available lower bound on what the architecture
costs. It also suggests a structural remedy in the spirit of the simplex
closure: predicting a correction to the analytically known initial field rather
than the field itself.

With that understood, \Cref{tab:sample269-time} corroborates
\Cref{fig:errorloc} on a single member. At $t=1$, before breakthrough, the
three models are indistinguishable, a spread under $1.5\%$ in the mean. At
$t=2$, as the bubble begins to pierce the interface, SEDONet is lowest on every
channel, by $37.0\%$ relative to the baseline and $19.3\%$ relative to FEDONet;
at $t=4$, with the bubble detached and the spike thin, the margins are $29.3\%$
and $10.6\%$.

At $t=3$ the ordering reverses on one channel, and the reversal is informative
rather than anomalous. FEDONet's mean is a few percent below SEDONet's, and the
margin comes almost entirely from $c_a$, while SEDONet retains $c_o$, $u$ and
$v$. At that instant the bubble is a compact, recently detached blob whose
boundary is a closed curve of high and nearly uniform curvature, in the
angular coordinate a genuinely periodic feature, which is what a Fourier
dictionary represents most efficiently, whereas the water spike beneath it
is a monotone, strongly non-periodic protrusion. A single member at a single
instant should not carry an architectural claim in either direction, but it
does show that the channel-wise ordering of \Cref{tab:test} is an aggregate
statement rather than a pointwise one, and it shows the two dictionaries
behaving as their priors predict.

The qualitative picture is in the end the more important result. All three
models reproduce the correct topology at every snapshot (\Cref{fig:fields}):
bubble detached, spike sharp, oil layer reconnected. \Cref{fig:fields-error}
shows that the absolute error is a thin shell hugging the bubble rim and the
plume flanks and is essentially black everywhere else, interface
localization error, the models placing the interface a fraction of a
cell off, rather than a field-level failure in which the bulk arrangement is
wrong. What distinguishes the models is the thickness and brightness of that
shell. The improvement sits exactly where the physics is.

\subsection{Physical diagnostics and constraint preservation}
\label{sec:res:physical}
A pointwise norm is not what a surrogate is ultimately used for; integral and
geometric quantities are, and a constraint-consistent output is, so we close by
evaluating both. The bubble kinematics follow the same pattern as the field error, and more
sharply. All three models track the centroid height closely until $t\approx1$,
after which the deviations grow monotonically; by the end of the horizon
SEDONet's deviation is $42\%$ below the baseline's and $33\%$ below FEDONet's,
both larger than the corresponding field-error gaps, and the interquartile
bands are well separated over the second half of the horizon. The improvement
in interface localization therefore translates into a better-than-proportional
improvement in the geometric quantity that localization determines.

The entrained volume is more equivocal and we report it as such. SEDONet's
deviation is the largest of the three until the plume has fully formed,
and only then becomes the smallest. Part of the early-time deviation is a
small-number artifact, since before breakthrough the entrained volume is a
near-zero integral over the diffuse tail of the interface profile and any
smoothing bias registers as a bias in the integral. But the effect is real, and
a surrogate used to predict entrainment specifically would need it addressed,
most naturally by supervising the diagnostic directly or by closing the initial
condition structurally.

\begin{figure}[H]
\centering
\includegraphics[width=0.92\textwidth]{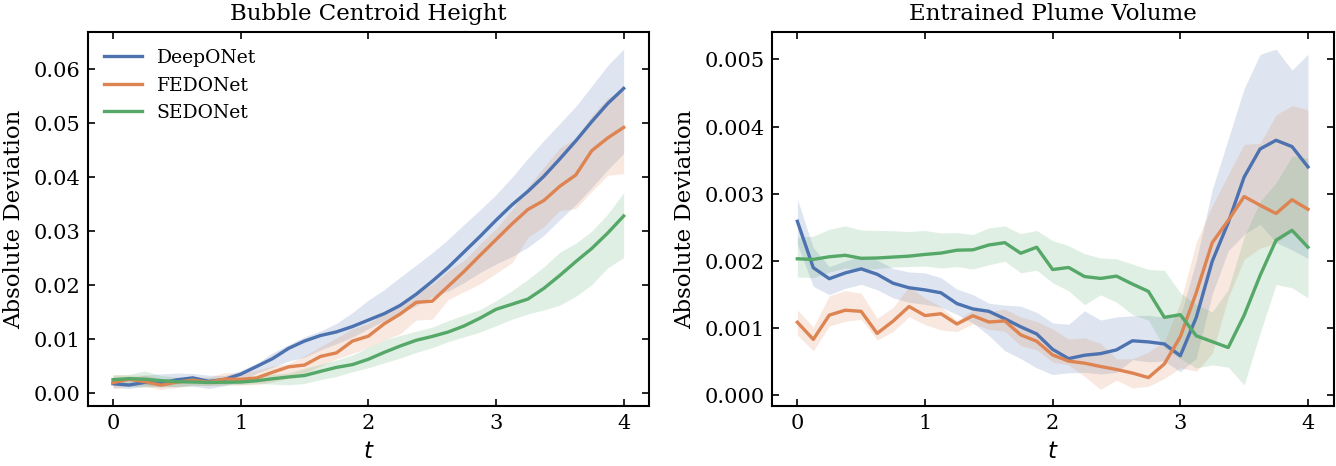}
\caption{Absolute deviation of two derived diagnostics computed from the
predicted fields and compared with the reference: bubble centroid height
$y_{\mathrm{bubble}}$ (left) and entrained plume volume $V_{\mathrm{plume}}$
(right). Lines are means over the $103$ test members, bands the interquartile
spread.}
\label{fig:diagfidelity}
\end{figure}

Constraint preservation is the one place where the surrogate is exactly right
rather than approximately so. With the closure of \cref{eq:closure-std} the
measured simplex, the violation sits at the single-precision rounding floor for all
three models, orders of magnitude below the drift present in the reference data
itself, and the value is identical across the three independently
trained networks. That identity is itself the diagnostic: a learned
approximation to the constraint would vary with architecture and with seed,
whereas a value reproduced by three different models is the rounding error of
the closure arithmetic and nothing else. It is worth being clear about what
this does not say. The surrogate satisfies the algebraic constraint exactly
while satisfying none of the differential ones, since neither incompressibility
nor mass conservation is enforced anywhere; structural constraints are
essentially free when they are affine, which is not the same as the surrogate
being physically consistent in general.
\section{Summary and Conclusions}
\label{sec:conclusion}

We have introduced a ternary Cahn-Hilliard-Navier-Stokes benchmark for
operator learning, an air bubble rising through water, piercing a water-oil
interface and entraining a plume into the oil, and used it to compare three
parameter-matched DeepONet variants that differ only in the trunk coordinate
embedding. The benchmark contributes three things its predecessors lack: a
five-channel space-time target rather than a single scalar field, an algebraic
constraint on the state that a faithful surrogate should respect exactly, and a
reference solver in which that constraint is a structural identity of the
discretization rather than an approximation, so that any violation exhibited by
a surrogate is provably its own.

On this benchmark the Chebyshev-embedded trunk reduces test relative $L^2$
error by $24.0\%$ relative to the raw-coordinate baseline and by $16.8\%$
relative to the Fourier-embedded variant, improves every output channel and the
bubble-centroid diagnostic, and preserves the simplex constraint to the
single-precision rounding floor through an affine closure. The comparison that
carries the claim is parameter-exact: the two embedded models have identical
parameter counts, depths, widths, optimizers and data, so the difference
between them is attributable to the coordinate dictionary and to nothing else.

\begin{table}[ht!]
\centering
\caption{Summary. All three models share branch and trunk depths, widths,
optimizer, schedule and data, and differ only in the trunk embedding $\phi$.
The FEDONet--SEDONet comparison is parameter-exact.}
\label{tab:summary}
\begin{tabular}{lccc}
\toprule
Model & Test rel.\ $L^2$ & Reduction & Parameters \\
\midrule
DeepONet (raw trunk) & 0.0821          & ---       & 0.73~M \\
FEDONet (Fourier)    & 0.0750          & $8.6\%$   & 0.86~M \\
SEDONet (Chebyshev)  & \textbf{0.0624} & $24.0\%$  & 0.86~M \\
\bottomrule
\end{tabular}
\end{table}

The interpretive claim we would defend is narrower than the headline number,
and the evidence for it is the localization of the gain rather than its size.
The advantage of the polynomial dictionary here is not boundary-layer
resolution: the error is smallest at the walls, at every time and for every
model, so although the Chebyshev prior does place its greatest amplitude near
the domain boundary, there is nothing there for it to resolve. What the
dictionary supplies instead is twofold. It represents the strongly
non-periodic, monotone trends of the solution in the wall-normal and temporal
coordinates in a handful of low-order modes, where a periodic dictionary needs
a slowly converging series; and it conditions the trunk input
deterministically, where the Fourier alternative achieves the same effect only
in expectation and at a bandwidth that must be chosen. Both advantages are
contingent on the problem: on a genuinely periodic domain, or one without a
dominant drift, we would not expect the ordering to hold.

Several limitations bear on how the numbers should be read. Each model was
trained once, so the reported gaps are point estimates without run-to-run error
bars, and repeated runs at matched budgets are the first thing this study
needs. The learning target is a diffuse-interface solution at a finite Cahn
number, so conclusions transfer to sharp-interface physics only insofar as the
diffuse solution approximates it. The design space contains a single
topological class by construction, so the surrogates have not been tested near
a bifurcation boundary, and all test members are drawn from that same design,
leaving extrapolation untested. The surrogate is also purely data driven apart
from the algebraic closure, and the configuration is two-dimensional, whereas
triple junctions in three dimensions are lines rather than points.
\section{Future Work}
\label{sec:future}

The most pressing extensions are evaluative rather than architectural. The
effect sizes reported here rest on one training run per model, so converting
them into defensible claims requires repetition at matched compute budgets
across several seeds and several choices of the Fourier bandwidth, so that a
tuned stochastic embedding and a deterministic one are compared on equal terms.
The design space should then be widened until it crosses a regime boundary, and
sampled outside the training box, since engineering use of a parametric
surrogate almost always involves some extrapolation. The same comparison
belongs in three dimensions, where triple junctions are lines rather than
points and the anisotropy of a truncated tensor-product dictionary becomes a
more delicate question, and on other multiphase systems, so that the conditions
under which a polynomial dictionary helps can be separated from the particulars
of one benchmark. Because the coordinate representation is modular, the same
question can be asked of the lifting layers of other neural-operator families.

The architectural directions follow from where the error still sits. All three
models spend their budget at the diffuse interfaces, which are precisely the
features a fixed global basis represents least efficiently; a dictionary
combining global trend modes with localized modes tied to the instantaneous
interface position would attack that budget directly, and more generally a
basis that adapts to the moving geometry is the natural route past the limits
of any fixed linear expansion. A second direction is to extend the structural
approach that worked for the algebraic constraint to the differential ones:
incompressibility and mass conservation are not affine, but they can be built
into the output representation or enforced by projection rather than left for
the data to teach, and the same logic applies to an initial condition known in
closed form. The benchmark itself is meant to be reusable, and a surrogate that
respects the differential structure as well as the algebraic one would be the
natural next entry.
\section*{Declaration of Competing Interest}
The authors declare that they have no known competing financial interests or personal relationships that could have appeared to influence the work reported in this paper.

\section*{Acknowledgements}
This work was supported by the grant in part by the AFOSR Grant FA9550-24-1-0327.

\section*{Data Availability}
The data supporting the findings of this study are available from the corresponding author upon reasonable request.

\bibliographystyle{unsrt}
\bibliography{references}

@article{cahn1958,
  title={Free energy of a nonuniform system. I. Interfacial free energy},
  author={Cahn, John W and Hilliard, John E},
  journal={The Journal of Chemical Physics},
  volume={28},
  number={2},
  pages={258--267},
  year={1958},
  publisher={American Institute of Physics},
  doi={10.1063/1.1744102}
}

@article{anderson1998,
  author  = {Anderson, D. M. and McFadden, G. B. and Wheeler, A. A.},
  title   = {Diffuse-Interface Methods in Fluid Mechanics},
  journal = {Annual Review of Fluid Mechanics},
  volume  = {30},
  pages   = {139--165},
  year    = {1998},
  doi     = {10.1146/annurev.fluid.30.1.139}
}

@article{lowengrub1998,
  author  = {Lowengrub, John and Truskinovsky, Lev},
  title   = {Quasi-Incompressible {C}ahn--{H}illiard Fluids and Topological Transitions},
  journal = {Proceedings of the Royal Society of London A},
  volume  = {454},
  number  = {1978},
  pages   = {2617--2654},
  year    = {1998},
  doi     = {10.1098/rspa.1998.0273}
}

@article{jacqmin1999,
  author  = {Jacqmin, David},
  title   = {Calculation of Two-Phase {N}avier--{S}tokes Flows Using Phase-Field Modeling},
  journal = {Journal of Computational Physics},
  volume  = {155},
  number  = {1},
  pages   = {96--127},
  year    = {1999},
  doi     = {10.1006/jcph.1999.6332}
}

@article{ding2007,
  author  = {Ding, Hang and Spelt, Peter D. M. and Shu, Chang},
  title   = {Diffuse Interface Model for Incompressible Two-Phase Flows with Large Density Ratios},
  journal = {Journal of Computational Physics},
  volume  = {226},
  number  = {2},
  pages   = {2078--2095},
  year    = {2007},
  doi     = {10.1016/j.jcp.2007.06.028}
}

@article{abels2012,
  author  = {Abels, Helmut and Garcke, Harald and Gr{\"u}n, G{\"u}nther},
  title   = {Thermodynamically Consistent, Frame Indifferent Diffuse Interface Models for Incompressible Two-Phase Flows with Different Densities},
  journal = {Mathematical Models and Methods in Applied Sciences},
  volume  = {22},
  number  = {3},
  pages   = {1150013},
  year    = {2012},
  doi     = {10.1142/S0218202511500138}
}

@article{degennes1985,
  author  = {de Gennes, P. G.},
  title   = {Wetting: Statics and Dynamics},
  journal = {Reviews of Modern Physics},
  volume  = {57},
  number  = {3},
  pages   = {827--863},
  year    = {1985},
  doi     = {10.1103/RevModPhys.57.827}
}

@article{boyer2006,
  author  = {Boyer, Franck and Lapuerta, C{\'e}line},
  title   = {Study of a Three Component {C}ahn--{H}illiard Flow Model},
  journal = {ESAIM: Mathematical Modelling and Numerical Analysis},
  volume  = {40},
  number  = {4},
  pages   = {653--687},
  year    = {2006},
  doi     = {10.1051/m2an:2006028}
}

@article{boyer2010,
  author  = {Boyer, Franck and Lapuerta, C{\'e}line and Minjeaud, Sebastian and Piar, Bruno and Quintard, Michel},
  title   = {{C}ahn--{H}illiard/{N}avier--{S}tokes Model for the Simulation of Three-Phase Flows},
  journal = {Transport in Porous Media},
  volume  = {82},
  number  = {3},
  pages   = {463--483},
  year    = {2010},
  doi     = {10.1007/s11242-009-9408-z}
}

@article{kim2005,
  author  = {Kim, Junseok and Lowengrub, John},
  title   = {Phase Field Modeling and Simulation of Three-Phase Flows},
  journal = {Interfaces and Free Boundaries},
  volume  = {7},
  number  = {4},
  pages   = {435--466},
  year    = {2005},
  doi     = {10.4171/IFB/132}
}

@article{kim2012,
  author  = {Kim, Junseok},
  title   = {Phase-Field Models for Multi-Component Fluid Flows},
  journal = {Communications in Computational Physics},
  volume  = {12},
  number  = {3},
  pages   = {613--661},
  year    = {2012},
  doi     = {10.4208/cicp.301110.040811a}
}

@article{hirt1981,
  author  = {Hirt, C. W. and Nichols, B. D.},
  title   = {Volume of Fluid ({VOF}) Method for the Dynamics of Free Boundaries},
  journal = {Journal of Computational Physics},
  volume  = {39},
  number  = {1},
  pages   = {201--225},
  year    = {1981},
  doi     = {10.1016/0021-9991(81)90145-5}
}

@article{sussman1994,
  author  = {Sussman, Mark and Smereka, Peter and Osher, Stanley},
  title   = {A Level Set Approach for Computing Solutions to Incompressible Two-Phase Flow},
  journal = {Journal of Computational Physics},
  volume  = {114},
  number  = {1},
  pages   = {146--159},
  year    = {1994},
  doi     = {10.1006/jcph.1994.1155}
}

@article{tryggvason2001,
  author  = {Tryggvason, G. and Bunner, B. and Esmaeeli, A. and Juric, D. and Al-Rawahi, N. and Tauber, W. and Han, J. and Nas, S. and Jan, Y.-J.},
  title   = {A Front-Tracking Method for the Computations of Multiphase Flow},
  journal = {Journal of Computational Physics},
  volume  = {169},
  number  = {2},
  pages   = {708--759},
  year    = {2001},
  doi     = {10.1006/jcph.2001.6726}
}

@article{smith2002,
  author  = {Smith, K. A. and Solis, F. J. and Chopp, D. L.},
  title   = {A Projection Method for Motion of Triple Junctions by Levels Sets},
  journal = {Interfaces and Free Boundaries},
  volume  = {4},
  number  = {3},
  pages   = {263--276},
  year    = {2002},
  doi     = {10.4171/IFB/61}
}

@article{brackbill1992,
  author  = {Brackbill, J. U. and Kothe, D. B. and Zemach, C.},
  title   = {A Continuum Method for Modeling Surface Tension},
  journal = {Journal of Computational Physics},
  volume  = {100},
  number  = {2},
  pages   = {335--354},
  year    = {1992},
  doi     = {10.1016/0021-9991(92)90240-Y}
}

@article{chorin1968,
  author  = {Chorin, Alexandre Joel},
  title   = {Numerical Solution of the {N}avier--{S}tokes Equations},
  journal = {Mathematics of Computation},
  volume  = {22},
  number  = {104},
  pages   = {745--762},
  year    = {1968},
  doi     = {10.1090/S0025-5718-1968-0242392-2}
}

@article{dodd2014,
  author  = {Dodd, Michael S. and Ferrante, Antonino},
  title   = {A Fast Pressure-Correction Method for Incompressible Two-Fluid Flows},
  journal = {Journal of Computational Physics},
  volume  = {273},
  pages   = {416--434},
  year    = {2014},
  doi     = {10.1016/j.jcp.2014.05.024}
}

@incollection{eyre1998,
  author    = {Eyre, David J.},
  title     = {Unconditionally Gradient Stable Time Marching the {C}ahn--{H}illiard Equation},
  booktitle = {Computational and Mathematical Models of Microstructural Evolution},
  series    = {Materials Research Society Symposium Proceedings},
  volume    = {529},
  pages     = {39--46},
  publisher = {Materials Research Society},
  address   = {Warrendale, PA},
  year      = {1998},
  doi       = {10.1557/PROC-529-39}
}

@article{shen2018sav,
  author  = {Shen, Jie and Xu, Jie and Yang, Jiang},
  title   = {The Scalar Auxiliary Variable ({SAV}) Approach for Gradient Flows},
  journal = {Journal of Computational Physics},
  volume  = {353},
  pages   = {407--416},
  year    = {2018},
  doi     = {10.1016/j.jcp.2017.10.021}
}

@article{bonhomme2012,
  author  = {Bonhomme, Roberto and Magnaudet, Jacques and Duval, Fabien and Piar, Bruno},
  title   = {Inertial Dynamics of Air Bubbles Crossing a Horizontal Fluid--Fluid Interface},
  journal = {Journal of Fluid Mechanics},
  volume  = {707},
  pages   = {405--443},
  year    = {2012},
  doi     = {10.1017/jfm.2012.288}
}

@article{gibou2019,
  author  = {Gibou, Frederic and Hyde, David and Fedkiw, Ron},
  title   = {Sharp Interface Approaches and Deep Learning Techniques for Multiphase Flows},
  journal = {Journal of Computational Physics},
  volume  = {380},
  pages   = {442--463},
  year    = {2019},
  doi     = {10.1016/j.jcp.2018.05.031}
}

@book{boyd2001,
  author    = {Boyd, John P.},
  title     = {Chebyshev and {F}ourier Spectral Methods},
  edition   = {2nd},
  publisher = {Dover Publications},
  address   = {Mineola, NY},
  year      = {2001}
}

@book{gottlieb1977,
  author    = {Gottlieb, David and Orszag, Steven A.},
  title     = {Numerical Analysis of Spectral Methods: Theory and Applications},
  publisher = {Society for Industrial and Applied Mathematics},
  address   = {Philadelphia, PA},
  year      = {1977},
  doi       = {10.1137/1.9781611970425}
}

@book{trefethen2013,
  author    = {Trefethen, Lloyd N.},
  title     = {Approximation Theory and Approximation Practice},
  publisher = {SIAM},
  address   = {Philadelphia, PA},
  year      = {2013}
}

@book{canuto2006,
  author    = {Canuto, Claudio and Hussaini, M. Youssuff and Quarteroni, Alfio and Zang, Thomas A.},
  title     = {Spectral Methods: Fundamentals in Single Domains},
  publisher = {Springer},
  address   = {Berlin, Heidelberg},
  year      = {2006},
  doi       = {10.1007/978-3-540-30726-6}
}

@article{chenchen1995,
  author  = {Chen, Tianping and Chen, Hong},
  title   = {Universal Approximation to Nonlinear Operators by Neural Networks with Arbitrary Activation Functions and Its Application to Dynamical Systems},
  journal = {IEEE Transactions on Neural Networks},
  volume  = {6},
  number  = {4},
  pages   = {911--917},
  year    = {1995},
  doi     = {10.1109/72.392253}
}

@article{lu2021,
  author  = {Lu, Lu and Jin, Pengzhan and Pang, Guofei and Zhang, Zhongqiang and Karniadakis, George Em},
  title   = {Learning Nonlinear Operators via {D}eep{ONet} Based on the Universal Approximation Theorem of Operators},
  journal = {Nature Machine Intelligence},
  volume  = {3},
  number  = {3},
  pages   = {218--229},
  year    = {2021},
  doi     = {10.1038/s42256-021-00302-5}
}

@article{kovachki2023,
  author  = {Kovachki, Nikola and Li, Zongyi and Liu, Burigede and Azizzadenesheli, Kamyar and Bhattacharya, Kaushik and Stuart, Andrew and Anandkumar, Anima},
  title   = {Neural Operator: Learning Maps Between Function Spaces with Applications to {PDE}s},
  journal = {Journal of Machine Learning Research},
  volume  = {24},
  number  = {89},
  pages   = {1--97},
  year    = {2023},
  doi = {10.5555/3648699.3648788}
}

@misc{li2021fno,
  author = {Li, Zongyi and Kovachki, Nikola and Azizzadenesheli, Kamyar and Liu, Burigede and Bhattacharya, Kaushik and Stuart, Andrew and Anandkumar, Anima},
  title  = {{F}ourier Neural Operator for Parametric Partial Differential Equations},
  year   = {2021},
  eprint = {2010.08895},
  archivePrefix = {arXiv},
  primaryClass  = {cs.LG},
  doi = {10.48550/arXiv.2010.08895}
}

@article{lanthaler2022,
  author  = {Lanthaler, Samuel and Mishra, Siddhartha and Karniadakis, George Em},
  title   = {Error Estimates for {D}eep{ONet}s: A Deep Learning Framework in Infinite Dimensions},
  journal = {Transactions of Mathematics and Its Applications},
  volume  = {6},
  number  = {1},
  pages   = {tnac001},
  year    = {2022},
  doi     = {10.1093/imatrm/tnac001}
}

@article{lu2022fair,
  author  = {Lu, Lu and Meng, Xuhui and Cai, Shengze and Mao, Zhiping and Goswami, Somdatta and Zhang, Zhongqiang and Karniadakis, George Em},
  title   = {A Comprehensive and Fair Comparison of Two Neural Operators (with Practical Extensions) Based on {FAIR} Data},
  journal = {Computer Methods in Applied Mechanics and Engineering},
  volume  = {393},
  pages   = {114778},
  year    = {2022},
  doi     = {10.1016/j.cma.2022.114778}
}

@article{tripura2023wno,
  author  = {Tripura, Tapas and Chakraborty, Souvik},
  title   = {Wavelet Neural Operator for Solving Parametric Partial Differential Equations in Computational Mechanics Problems},
  journal = {Computer Methods in Applied Mechanics and Engineering},
  volume  = {404},
  pages   = {115783},
  year    = {2023},
  doi     = {10.1016/j.cma.2022.115783}
}

@inproceedings{hao2023gnot,
  author    = {Hao, Zhongkai and Wang, Zhengyi and Su, Hang and Ying, Chengyang and Dong, Yinpeng and Liu, Songming and Cheng, Ze and Song, Jian and Zhu, Jun},
  title     = {{GNOT}: A General Neural Operator Transformer for Operator Learning},
  booktitle = {International Conference on Machine Learning (ICML)},
  pages     = {12556--12569},
  year      = {2023},
  doi = {10.48550/arXiv.2302.14376}
}

@inproceedings{raonic2023cno,
  author    = {Raonic, Bogdan and Molinaro, Roberto and Rohner, Tim and Mishra, Siddhartha and de B{\'e}zenac, Emmanuel},
  title     = {Convolutional Neural Operators for Robust and Accurate Learning of {PDE}s},
  booktitle = {Advances in Neural Information Processing Systems (NeurIPS)},
  year      = {2023},
  doi = {10.48550/arXiv.2302.01178}
}

@article{li2023geofno,
  author  = {Li, Zongyi and Huang, Daniel Zhengyu and Liu, Burigede and Anandkumar, Anima},
  title   = {{F}ourier Neural Operator with Learned Deformations for {PDE}s on General Geometries},
  journal = {Journal of Machine Learning Research},
  volume  = {24},
  number  = {388},
  pages   = {1--26},
  year    = {2023},
  doi = {10.5555/3648699.3649087}
}

@article{fanaskov2023spectral,
  author  = {Fanaskov, V. S. and Oseledets, I. V.},
  title   = {Spectral Neural Operators},
  journal = {Doklady Mathematics},
  volume  = {108},
  pages   = {S226--S232},
  year    = {2023},
  doi     = {10.1134/S1064562423701107}
}

@article{liu2024opno,
  author  = {Liu, Ziyuan and Wang, Haifeng and Zhang, Hong and Bao, Kaijun and Qian, Xu and Song, Songhe},
  title   = {Render Unto Numerics: Orthogonal Polynomial Neural Operator for {PDE}s with Nonperiodic Boundary Conditions},
  journal = {SIAM Journal on Scientific Computing},
  volume  = {46},
  number  = {4},
  pages   = {C323--C348},
  year    = {2024},
  doi     = {10.1137/23M1556320}
}

@article{wu2024ionet,
  author  = {Wu, Sidi and Zhu, Aiqing and Tang, Yifa and Lu, Benzhuo},
  title   = {Solving Parametric Elliptic Interface Problems via Interfaced Operator Network},
  journal = {Journal of Computational Physics},
  volume  = {514},
  pages   = {113217},
  year    = {2024},
  doi     = {10.1016/j.jcp.2024.113217}
}

@article{bi2025xideeponet,
  author  = {Bi, Xiaoxue and Chen, Xin and Zhao, Chuang and Li, Qiang and Zhang, Jian},
  title   = {{XI}-{D}eep{ONet}: An Operator Learning Method for Elliptic Interface Problems},
  journal = {Journal of Computational Physics},
  volume  = {538},
  pages   = {114164},
  year    = {2025},
  doi     = {10.1016/j.jcp.2025.114164}
}

@article{bahmani2025rino,
  author  = {Bahmani, Bahador and Goswami, Somdatta and Kevrekidis, Ioannis G. and Shields, Michael D.},
  title   = {A Resolution Independent Neural Operator},
  journal = {Computer Methods in Applied Mechanics and Engineering},
  volume  = {444},
  pages   = {118113},
  year    = {2025},
  doi     = {10.1016/j.cma.2025.118113}
}

@article{howard2023multifidelity,
  author  = {Howard, Amanda A. and Perego, Mauro and Karniadakis, George Em and Stinis, Panos},
  title   = {Multifidelity Deep Operator Networks for Data-Driven and Physics-Informed Problems},
  journal = {Journal of Computational Physics},
  volume  = {493},
  pages   = {112462},
  year    = {2023},
  doi     = {10.1016/j.jcp.2023.112462}
}

@article{liu2024spectralbias,
  author  = {Liu, Xinliang and Xu, Bo and Cao, Shuhao and Zhang, Lei},
  title   = {Mitigating Spectral Bias for the Multiscale Operator Learning},
  journal = {Journal of Computational Physics},
  volume  = {506},
  pages   = {112944},
  year    = {2024},
  doi     = {10.1016/j.jcp.2024.112944}
}

@article{raissi2019pinn,
  author  = {Raissi, M. and Perdikaris, P. and Karniadakis, G. E.},
  title   = {Physics-Informed Neural Networks: A Deep Learning Framework for Solving Forward and Inverse Problems Involving Nonlinear Partial Differential Equations},
  journal = {Journal of Computational Physics},
  volume  = {378},
  pages   = {686--707},
  year    = {2019},
  doi     = {10.1016/j.jcp.2018.10.045}
}

@article{wang2021pideeponet,
  author  = {Wang, Sifan and Wang, Hanwen and Perdikaris, Paris},
  title   = {Learning the Solution Operator of Parametric Partial Differential Equations with Physics-Informed {D}eep{ONet}s},
  journal = {Science Advances},
  volume  = {7},
  number  = {40},
  pages   = {eabi8605},
  year    = {2021},
  doi     = {10.1126/sciadv.abi8605}
}

@article{goswami2022vdeeponet,
  author  = {Goswami, Somdatta and Yin, Minglang and Yu, Yue and Karniadakis, George Em},
  title   = {A Physics-Informed Variational {D}eep{ONet} for Predicting Crack Path in Quasi-Brittle Materials},
  journal = {Computer Methods in Applied Mechanics and Engineering},
  volume  = {391},
  pages   = {114587},
  year    = {2022},
  doi     = {10.1016/j.cma.2022.114587}
}

@article{karniadakis2021,
  author  = {Karniadakis, George Em and Kevrekidis, Ioannis G. and Lu, Lu and Perdikaris, Paris and Wang, Sifan and Yang, Liu},
  title   = {Physics-Informed Machine Learning},
  journal = {Nature Reviews Physics},
  volume  = {3},
  pages   = {422--440},
  year    = {2021},
  doi     = {10.1038/s42254-021-00314-5}
}

@misc{rahaman2019,
  author = {Rahaman, Nasim and Baratin, Aristide and Arpit, Devansh and Draxler, Felix and Lin, Min and Hamprecht, Fred A. and Bengio, Yoshua and Courville, Aaron},
  title  = {On the Spectral Bias of Neural Networks},
  year   = {2019},
  eprint = {1806.08734},
  archivePrefix = {arXiv},
  primaryClass  = {stat.ML},
  doi = {10.48550/arXiv.1806.08734}
}

@article{xu2020frequency,
  author  = {Xu, Zhi-Qin John and Zhang, Yaoyu and Luo, Tao and Xiao, Yanyang and Ma, Zheng},
  title   = {Frequency Principle: {F}ourier Analysis Sheds Light on Deep Neural Networks},
  journal = {Communications in Computational Physics},
  volume  = {28},
  number  = {5},
  pages   = {1746--1767},
  year    = {2020},
  doi     = {10.4208/cicp.OA-2020-0085}
}

@inproceedings{jacot2018ntk,
  author    = {Jacot, Arthur and Gabriel, Franck and Hongler, Cl{\'e}ment},
  title     = {Neural Tangent Kernel: Convergence and Generalization in Neural Networks},
  booktitle = {Advances in Neural Information Processing Systems (NeurIPS)},
  year      = {2018},
  doi = {10.48550/arXiv.1806.07572}
}

@inproceedings{rahimi2007,
  author    = {Rahimi, Ali and Recht, Benjamin},
  title     = {Random Features for Large-Scale Kernel Machines},
  booktitle = {Advances in Neural Information Processing Systems (NeurIPS)},
  volume    = {20},
  year      = {2007}
}

@misc{tancik2020,
  author = {Tancik, Matthew and Srinivasan, Pratul P. and Mildenhall, Ben and Fridovich-Keil, Sara and Raghavan, Nithin and Singhal, Utkarsh and Ramamoorthi, Ravi and Barron, Jonathan T. and Ng, Ren},
  title  = {{F}ourier Features Let Networks Learn High Frequency Functions in Low Dimensional Domains},
  year   = {2020},
  eprint = {2006.10739},
  archivePrefix = {arXiv},
  primaryClass  = {cs.CV},
  doi = {10.48550/arXiv.2006.10739}
}

@article{wang2021eigenvector,
  author  = {Wang, Sifan and Wang, Hanwen and Perdikaris, Paris},
  title   = {On the Eigenvector Bias of {F}ourier Feature Networks: From Regression to Solving Multi-Scale {PDE}s with Physics-Informed Neural Networks},
  journal = {Computer Methods in Applied Mechanics and Engineering},
  volume  = {384},
  pages   = {113938},
  year    = {2021},
  doi     = {10.1016/j.cma.2021.113938}
}

@article{sojitra2025,
  title={FEDONet: Fourier-embedded DeepONet for spectrally accurate operator learning},
  author={Sojitra, Arth and Dhingra, Mrigank and San, Omer},
  journal={Journal of Computational Physics},
  pages={114931},
  year={2026},
  publisher={Elsevier},
  doi = {10.48550/arXiv.2509.12344
}
}

@misc{abid2026,
  author = {Abid, Muhammad and San, Omer},
  title  = {Spectral Embedding via {C}hebyshev Bases for Robust {D}eep{ONet} Approximation},
  year   = {2026},
  eprint = {2512.09165},
  archivePrefix = {arXiv},
  primaryClass  = {cs.LG},
  doi = {10.48550/arXiv.2512.09165}
}

@misc{yin2024cheb,
  author = {Yin, Pengfei and Ling, Shuai and Ying, Wenjun},
  title  = {{C}hebyshev Spectral Neural Networks for Solving Partial Differential Equations},
  year   = {2024},
  eprint = {2407.03347},
  archivePrefix = {arXiv},
  primaryClass  = {math.NA},
  doi = {10.48550/arXiv.2407.03347}
}

@misc{xu2024chebfeature,
  author = {Xu, Zhongshu and Chen, Yuan and Xiu, Dongbin},
  title  = {{C}hebyshev Feature Neural Network for Accurate Function Approximation},
  year   = {2024},
  eprint = {2409.19135},
  archivePrefix = {arXiv},
  primaryClass  = {cs.LG},
  doi = {10.48550/arXiv.2409.19135}
}

@article{brunton2020,
  author  = {Brunton, Steven L. and Noack, Bernd R. and Koumoutsakos, Petros},
  title   = {Machine Learning for Fluid Mechanics},
  journal = {Annual Review of Fluid Mechanics},
  volume  = {52},
  pages   = {477--508},
  year    = {2020},
  doi     = {10.1146/annurev-fluid-010719-060214}
}

@article{vinuesa2022,
  author  = {Vinuesa, Ricardo and Brunton, Steven L.},
  title   = {Enhancing Computational Fluid Dynamics with Machine Learning},
  journal = {Nature Computational Science},
  volume  = {2},
  pages   = {358--366},
  year    = {2022},
  doi     = {10.1038/s43588-022-00264-7}
}

@article{kochkov2021,
  author  = {Kochkov, Dmitrii and Smith, Jamie A. and Alieva, Ayya and Wang, Qing and Brenner, Michael P. and Hoyer, Stephan},
  title   = {Machine Learning-Accelerated Computational Fluid Dynamics},
  journal = {Proceedings of the National Academy of Sciences},
  volume  = {118},
  number  = {21},
  pages   = {e2101784118},
  year    = {2021},
  doi     = {10.1073/pnas.2101784118}
}

@article{wen2022ufno,
  author  = {Wen, Gege and Li, Zongyi and Azizzadenesheli, Kamyar and Anandkumar, Anima and Benson, Sally M.},
  title   = {{U}-{FNO}---An Enhanced {F}ourier Neural Operator-Based Deep-Learning Model for Multiphase Flow},
  journal = {Advances in Water Resources},
  volume  = {163},
  pages   = {104180},
  year    = {2022},
  doi     = {10.1016/j.advwatres.2022.104180}
}

@article{wen2023nested,
  author  = {Wen, Gege and Li, Zongyi and Long, Qirui and Azizzadenesheli, Kamyar and Anandkumar, Anima and Benson, Sally M.},
  title   = {Real-Time High-Resolution {CO$_2$} Geological Storage Prediction Using Nested {F}ourier Neural Operators},
  journal = {Energy \& Environmental Science},
  volume  = {16},
  pages   = {1732--1741},
  year    = {2023},
  doi     = {10.1039/D2EE04204E}
}

@misc{pathak2022,
  author = {Pathak, Jaideep and Subramanian, Shashank and Harrington, Peter and Raja, Sanjeev and Chattopadhyay, Ashesh and Mardani, Morteza and Kurth, Thorsten and Hall, David and Li, Zongyi and Azizzadenesheli, Kamyar and Hassanzadeh, Pedram and Kashinath, Karthik and Anandkumar, Animashree},
  title  = {{F}our{C}ast{N}et: A Global Data-Driven High-Resolution Weather Model Using Adaptive {F}ourier Neural Operators},
  year   = {2022},
  eprint = {2202.11214},
  archivePrefix = {arXiv},
  primaryClass  = {physics.ao-ph},
  doi = {10.48550/arXiv.2202.11214}
}

@article{lam2023graphcast,
  author  = {Lam, Remi and Sanchez-Gonzalez, Alvaro and Willson, Matthew and others},
  title   = {Learning Skillful Medium-Range Global Weather Forecasting},
  journal = {Science},
  volume  = {382},
  number  = {6677},
  pages   = {1416--1421},
  year    = {2023},
  doi     = {10.1126/science.adi2336}
}

@inproceedings{takamoto2022pdebench,
  author    = {Takamoto, Makoto and Praditia, Timothy and Leiteritz, Raphael and MacKinlay, Dan and Alesiani, Francesco and Pfl{\"u}ger, Dirk and Niepert, Mathias},
  title     = {{PDEB}ench: An Extensive Benchmark for Scientific Machine Learning},
  booktitle = {Advances in Neural Information Processing Systems (NeurIPS) Datasets and Benchmarks Track},
  year      = {2022}
}

@inproceedings{herde2024poseidon,
  author    = {Herde, Maximilian and Raoni{\'c}, Bogdan and Rohner, Tobias and K{\"a}ppeli, Roger and Molinaro, Roberto and de B{\'e}zenac, Emmanuel and Mishra, Siddhartha},
  title     = {Poseidon: Efficient Foundation Models for {PDE}s},
  booktitle = {Advances in Neural Information Processing Systems (NeurIPS)},
  volume    = {37},
  year      = {2024},
  doi       = {10.52202/079017-2311}
}

@article{sobol1967,
  author  = {Sobol', I. M.},
  title   = {On the Distribution of Points in a Cube and the Approximate Evaluation of Integrals},
  journal = {USSR Computational Mathematics and Mathematical Physics},
  volume  = {7},
  number  = {4},
  pages   = {86--112},
  year    = {1967},
  doi     = {10.1016/0041-5553(67)90144-9}
}

@inproceedings{kingma2015adam,
  author    = {Kingma, Diederik P. and Ba, Jimmy},
  title     = {{A}dam: A Method for Stochastic Optimization},
  booktitle = {International Conference on Learning Representations (ICLR)},
  year      = {2015}
}

@misc{hendrycks2016gelu,
  author = {Hendrycks, Dan and Gimpel, Kevin},
  title  = {{G}aussian Error Linear Units ({GELU}s)},
  year   = {2016},
  eprint = {1606.08415},
  archivePrefix = {arXiv},
  primaryClass  = {cs.LG},
  doi = {10.48550/arXiv.1606.08415}
}

@misc{jax2018,
  author = {Bradbury, James and Frostig, Roy and Hawkins, Peter and Johnson, Matthew James and Leary, Chris and Maclaurin, Dougal and Necula, George and Paszke, Adam and VanderPlas, Jake and Wanderman-Milne, Skye and Zhang, Qiao},
  title  = {{JAX}: Composable Transformations of {P}ython$+${N}um{P}y Programs},
  year   = {2018},
  howpublished = {\url{https://github.com/jax-ml/jax}}
}

@article{beucler2021,
  author  = {Beucler, Tom and Pritchard, Michael and Rasp, Stephan and Ott, Jordan and Baldi, Pierre and Gentine, Pierre},
  title   = {Enforcing Analytic Constraints in Neural Networks Emulating Physical Systems},
  journal = {Physical Review Letters},
  volume  = {126},
  number  = {9},
  pages   = {098302},
  year    = {2021},
  doi     = {10.1103/PhysRevLett.126.098302}
}

@article{cohen2015,
  author  = {Cohen, Albert and DeVore, Ronald},
  title   = {Approximation of High-Dimensional Parametric {PDE}s},
  journal = {Acta Numerica},
  volume  = {24},
  pages   = {1--159},
  year    = {2015},
  doi     = {10.1017/S0962492915000033}
}

@inproceedings{ohlberger2016,
  author    = {Ohlberger, Mario and Rave, Stephan},
  title     = {Reduced Basis Methods: Success, Limitations and Future Challenges},
  booktitle = {Proceedings of the Conference Algoritmy},
  pages     = {1--12},
  year      = {2016},
  eprint    = {1511.02021},
  archivePrefix = {arXiv},
  primaryClass  = {math.NA},
  doi       = {10.48550/arXiv.1511.02021}
}

@article{peherstorfer2022,
  author  = {Peherstorfer, Benjamin},
  title   = {Breaking the {K}olmogorov Barrier with Nonlinear Model Reduction},
  journal = {Notices of the American Mathematical Society},
  volume  = {69},
  number  = {5},
  pages   = {725--733},
  year    = {2022},
  doi     = {10.1090/noti2475}
}

@article{lee2020,
  author  = {Lee, Kookjin and Carlberg, Kevin T.},
  title   = {Model Reduction of Dynamical Systems on Nonlinear Manifolds Using Deep Convolutional Autoencoders},
  journal = {Journal of Computational Physics},
  volume  = {404},
  pages   = {108973},
  year    = {2020},
  doi     = {10.1016/j.jcp.2019.108973}
}

\appendix

\section{Proof of Proposition~\ref{prop:simplex}}
\label{app:proof}

With $M_i = M_0/\Sigma_i$, sum \cref{eq:mu} against $M_i$:
\begin{equation}
\sum_{i=1}^{3} M_i \mu_i
= M_0\left[
\frac{12}{\varepsilon}\sum_{i}\frac{1}{\Sigma_i}\frac{\partial F_0}{\partial c_i}
-\frac{4\Sigma_T}{\varepsilon}\left(\sum_{i}\frac{1}{\Sigma_i}\right)
 \sum_{j}\frac{1}{\Sigma_j}\frac{\partial F_0}{\partial c_j}
-\frac{3}{4}\varepsilon \sum_{i}\nabla^2 c_i
\right].
\end{equation}
By \cref{eq:spreading}, $\sum_i \Sigma_i^{-1} = 3/\Sigma_T$, so the second
bracketed term equals $\tfrac{12}{\varepsilon}\sum_j
\Sigma_j^{-1}\partial_{c_j}F_0$ and cancels the first exactly; the third
vanishes whenever \cref{eq:simplex} holds. Hence $\sum_i M_i\mu_i \equiv 0$.
Summing \cref{eq:ch} over $i$ and using $\nabla\!\cdot\vu = 0$,
\begin{equation}
\partial_t\Big(\sum_i c_i\Big) + \vu\!\cdot\!\nabla\Big(\sum_i c_i\Big)
= \nabla^2\Big(\sum_i M_i \mu_i\Big) = 0 ,
\end{equation}
so $\sum_i c_i \equiv 1$ is propagated exactly. The cancellation carries to the
discrete level provided the same discrete Laplacian is used in \cref{eq:mu} and
\cref{eq:ch}, which holds for the scheme of \Cref{app:dct}; the discrete
statement is \Cref{prop:simplex-discrete}. $\blacksquare$

\section{Spectral properties of the Chebyshev dictionary}
\label{app:cheb}

This appendix sources the conditioning claim of \Cref{sec:opnet:compare} and
reports the deviation from it on the grid actually used.

\subsection{Exact orthogonality}
\label{app:cheb:exact}

The Chebyshev polynomials satisfy
\begin{equation}
\int_{-1}^{1}\frac{T_m(\xi)T_n(\xi)}{\sqrt{1-\xi^2}}\,d\xi =
\begin{cases}
0, & m\neq n,\\
\pi, & m = n = 0,\\
\pi/2, & m = n \ge 1 ,
\end{cases}
\label{eq:cheb-orth}
\end{equation}
and tensorization preserves this: with
$w_3(\bs{\xi}) = \prod_{d}(1-\xi_d^2)^{-1/2}$,
\begin{equation}
\int_{[-1,1]^3} \Phi_{klm}\,\Phi_{k'l'm'}\; w_3\, d\bs{\xi}
= \gamma_k\gamma_l\gamma_m\,\delta_{kk'}\delta_{ll'}\delta_{mm'},
\qquad \gamma_0 = \pi,\; \gamma_{n\ge1} = \pi/2 .
\label{eq:cheb-orth-3d}
\end{equation}
The discrete counterpart is what the embedding actually sees. On an $m$-point
Gauss--Chebyshev set $\xi_q = \cos\big((q+\tfrac12)\pi/m\big)$ the relation
\begin{equation}
\frac{1}{m}\sum_{q=1}^{m} T_j(\xi_q)T_n(\xi_q) = \gamma_n\,\delta_{jn},
\qquad \gamma_0 = 1,\quad \gamma_{n\ge1} = \tfrac12 ,
\label{eq:cheb-disc}
\end{equation}
holds exactly for $j,n < m$, so on the tensor product of such sets the
empirical Gram matrix
\begin{equation}
G = \frac{1}{Q}\sum_{q=1}^{Q}
    \phi(\vzeta^{q})\,\phi(\vzeta^{q})^{\!\top}
\label{eq:gram-app}
\end{equation}
is diagonal with entries
$\gamma_k\gamma_l\gamma_m \in \{1,\tfrac12,\tfrac14,\tfrac18\}$. The extreme
values are attained by $\Phi_{000}$ and by any mode with all three indices
nonzero, both of which the graded set $\mathcal{A}$ contains, so
\begin{equation}
\kappa(G) = \frac{1}{1/8} = 8
\label{eq:kappa-app}
\end{equation}
independently of $d_{\mathrm{trunk}}$, of $K$, and of any random draw, and
rescaling each mode by $\gamma_\alpha^{-1/2}$ gives $\kappa(G) = 1$ exactly.
We confirm \cref{eq:kappa-app} numerically: with $K = 10$ the off-diagonal
entries of the one-dimensional Gram matrix are at most $1.8\times10^{-16}$ and
$\kappa = 2$ per axis, hence $2^3 = 8$ in three dimensions.

\subsection{Conditioning on the uniform query grid}
\label{app:cheb:uniform}

The query set is sampled uniformly rather than under the weight $w(\xi)$, so
\cref{eq:cheb-disc} does not apply. Evaluated directly with $K = 10$, the
per-axis condition numbers are $15.1$, $13.7$ and $5.2$, and the graded
$256$-mode dictionary gives
\begin{equation}
\kappa(G) \approx 2.7\times10^{2}, \qquad
\kappa(G) \approx 74 \ \text{after diagonal rescaling},
\label{eq:kappa-uniform}
\end{equation}
against $8$ and $1$ under the orthogonality weight. This is a property of the
sampling measure, not of the resolution: refining the grid drives the per-axis
figure to $13.51$ rather than to $2$, since
$\int_{-1}^{1}T_mT_n\,d\xi \ne 0$ in general (for instance
$\int T_0T_2\,d\xi = -\tfrac23$).

Two consequences. First, \cref{eq:kappa-app} is the orthogonal-weight value
rather than the value in force during training, the realized whitening being
weaker by about a factor of thirty. What the argument of
\Cref{sec:opnet:compare} needs is untouched: a Gram matrix that is fixed,
deterministic, seed-identical and bandwidth-free, which
\cref{eq:kappa-uniform} still is, where \Cref{app:rff} shows the Fourier
alternative varies over twenty orders of magnitude with a parameter that must
be chosen in advance. Second, normalized Legendre polynomials
$\tilde P_n = \sqrt{2n+1}\,P_n$ are orthogonal under the uniform measure and
give $\kappa(G) = 1.15$ per axis, $1.53$ in three dimensions. Nothing in our
mechanism selects Chebyshev over Legendre within the polynomial family, and the
comparison is untested here.

\section{Spectral properties of the random Fourier dictionary}
\label{app:rff}

The claims made for the Fourier embedding are comparative, so they need the
same treatment as \Cref{app:cheb}.

\subsection{Induced kernel}
\label{app:rff:kernel}

With $B_{ij}\sim\mathcal{N}(0,\sigma_B^2)$, the identity
$\sin a\sin b + \cos a\cos b = \cos(a-b)$ gives
\begin{equation}
\frac{1}{M}
\phi_{\mathrm{F}}(\boldsymbol{\zeta})^{\!\top}
\phi_{\mathrm{F}}(\boldsymbol{\zeta}')
=
\frac{1}{M}\sum_{j=1}^{M}
\cos\!\Big(
2\pi \mathbf{b}_j^{\!\top}
(\boldsymbol{\xi}-\boldsymbol{\xi}')
\Big)
\;\xrightarrow[M\to\infty]{}\;
\exp\!\Big(
-2\pi^2\sigma_B^2
\big\|\boldsymbol{\xi}-\boldsymbol{\xi}'\big\|^2
\Big),
\label{eq:rff-kernel-app}
\end{equation}
a Gaussian kernel of length scale $(2\pi\sigma_B)^{-1}$ \cite{rahimi2007}.
Three properties contrast with the polynomial dictionary. It is stationary, so
the prior is identical at the walls and in the interior, where the polynomial
prior is maximal at the corners of $[-1,1]^3$; on this problem that difference
is inert, since the error is smallest at the walls for every model
(\Cref{sec:res:localization}). It is isotropic, so one scalar sets the
resolvable scale on all three axes, a strong assumption for a field whose
structure in $t$ is a drift and in $x$ a thin neck, and the reason
\cref{eq:affine} is a fairness condition rather than a convention. And it is
built from functions periodic on the embedding scale, which is the mechanism
quantified in \Cref{app:rff:trend}.

\subsection{Conditioning and its bandwidth dependence}
\label{app:rff:cond}

Distinct frequency rows are uncorrelated and $\mathbb{E}[\sin^2] =
\mathbb{E}[\cos^2] = \tfrac12$, so $\mathbb{E}[G] = \tfrac12 I$: the dictionary
is whitened \emph{in expectation}. What the optimizer sees is the realized $G$,
which depends sharply on $\sigma_B$ (\Cref{tab:rff-cond}).

\begin{table}[ht!]
\centering
\caption{Conditioning of the Fourier feature Gram matrix on the query grid
($M = 128$). The graded Chebyshev dictionary on the same grid gives
$\kappa(G)\approx2.7\times10^{2}$ with no parameter to select.}
\label{tab:rff-cond}
\begin{tabular}{lccccccc}
\toprule
$\sigma_B$ & $0.05$ & $0.25$ & $0.5$ & $1.0$ & $2.0$ & $4.0$ & $8.0$ \\
\midrule
$\kappa(G)$ & $2\times10^{20}$ & $5\times10^{15}$ & $5\times10^{8}$
            & $3.1\times10^{2}$ & $3.4$ & $1.5$ & $2.3$ \\
\bottomrule
\end{tabular}
\end{table}

As $\sigma_B\to0$ the features collapse toward rank one and $G$ becomes
singular to single precision; near $\sigma_B = 4$ the dictionary is better
conditioned than the polynomial one by two orders of magnitude. Both facts
together are the point: the polynomial dictionary is not uniformly better
conditioned, it is \emph{unconditionally} conditioned, its $\kappa$ being a
fixed number known before training where the Fourier value sweeps twenty orders
of magnitude across a parameter chosen without knowing the answer. Seed
dependence compounds this: at $\sigma_B = 1$, six draws of $B$ give $\kappa$
between $3.1\times10^{2}$ and $4.0\times10^{3}$. The polynomial dictionary is
bit-identical across seeds, so the two embeddings should not be expected to
have comparable seed variance.

\subsection{The cost of a monotone trend}
\label{app:rff:trend}

\Cref{sec:res:accuracy} attributes the divergence between the embeddings on the
large-region concentration fields to trend representation. Fitting
$f(\xi) = \xi$ in least squares on the wall-normal grid, the first two
Chebyshev modes give a relative residual of $6\times10^{-17}$ --- exact, since
$T_1 = \xi$. Two Fourier features at the same budget give $6.9\times10^{-3}$ at
$\sigma_B = 0.5$, $2.8\times10^{-2}$ at $1$, $1.2\times10^{-1}$ at $2$ and
$6.3\times10^{-1}$ at $4$; reaching $10^{-3}$ takes $8$ features at
$\sigma_B = 1$ and $16$ at $\sigma_B = 2$. The temporal grid gives the same
figures to two digits.

The residual grows with bandwidth while the conditioning of
\Cref{tab:rff-cond} improves with it, so within the Fourier family no single
$\sigma_B$ serves both and the trade-off is structural rather than a matter of
tuning. The magnitudes are a few percent of the trend, which is the right scale
to explain the per-channel gap of \Cref{tab:reductions}; it would not explain
an order of magnitude, and we do not claim one.

\section{The DCT-II Poisson and Helmholtz solves}
\label{app:dct}

On a uniform cell-centred grid $N_x\times N_y$ with spacing $h$ and homogeneous
Neumann conditions, the five-point Laplacian is diagonalized by the type-II
DCT: with $\widehat{\cdot}$ the two-dimensional DCT-II,
\begin{equation}
\widehat{(\nabla_h^2 f)}_{mn} = \lambda_{mn}\,\widehat{f}_{mn},
\qquad
\lambda_{mn} = \frac{2}{h^2}
\left[\cos\!\Big(\frac{\pi m}{N_x}\Big)
    + \cos\!\Big(\frac{\pi n}{N_y}\Big) - 2\right].
\label{eq:dct-eig}
\end{equation}
Both the pressure Poisson solve and the linear implicit part of the
Cahn--Hilliard update reduce to elementwise division by a known function of
$\lambda_{mn}$, with the $\lambda_{00} = 0$ mode handled by fixing the mean, at
$\mathcal{O}(N_xN_y\log N_xN_y)$ per solve and with no iteration. The implicit
Cahn--Hilliard operator is a polynomial in $\nabla_h^2$
(\cref{eq:ch-scheme}), so the same transform pair serves both and each species
costs two DCTs per step.

\section{The benchmark dataset}
\label{app:data}

This appendix specifies the reference solver and the ensemble. With
\Cref{tab:fixed,tab:sobol} it reproduces the benchmark completely.

\subsection{Solver}
\label{app:data:solver}

The discretization is staggered (MAC) finite differences on a uniform grid,
second-order centred in space, first order in time, explicit for momentum and
for the Cahn--Hilliard nonlinearity and implicit for its linear part. Three
choices make an ensemble of this size affordable while preserving the structure
of \Cref{sec:model:simplex}.

The first exploits the boundary conditions. Under \cref{eq:bc} every scalar
satisfies a homogeneous Neumann condition, so the discrete Laplacian is
diagonalized exactly by the DCT-II (\Cref{app:dct}) and both the pressure
Poisson equation and the linear Cahn--Hilliard update become one-shot solves,
with no iteration, no tolerance and no run-to-run variability in cost.

The second scales the mobility and stabilization so that all three species
share one implicit operator. With $M_i = M_0/\Sigma_i$ and $S_i = S_0\Sigma_i$,
both $M_i\Sigma_i = M_0$ and $M_iS_i = M_0S_0$ are species-independent, so
\begin{equation}
\Big(\mathcal{I} - \Delta t\,M_0 S_0 \nabla_h^2
      + \tfrac{3}{4}\Delta t\,M_0\,\varepsilon\,\nabla_h^4\Big) c_i^{\,n+1}
= \Big(\mathcal{I} - \Delta t\,M_0 S_0 \nabla_h^2\Big) c_i^{\,n}
 + \Delta t\Big[-\nabla_h\!\cdot\!\big(\vu^{\,n} c_i^{\,n}\big)
 + M_i \nabla_h^2 N_i^{\,n}\Big]
\label{eq:ch-scheme}
\end{equation}
has the same left-hand operator for every phase. This is what carries
\Cref{prop:simplex} to the discrete level.

\begin{proposition}[Discrete simplex preservation]
\label{prop:simplex-discrete}
Let $M_i = M_0/\Sigma_i$ and $S_i = S_0\Sigma_i$, let the face-interpolated
phase fields in the advective flux be linear in $\vc$, and let
$\nabla_h\!\cdot\!\vu^{\,n} = 0$. If $\sum_i c_i^{\,n} = 1$ pointwise, then
$\sum_i c_i^{\,n+1} = 1$ pointwise.
\end{proposition}

\noindent\textit{Proof.}
Sum \cref{eq:ch-scheme} over $i$. The left operator and the first right
operator are species-independent, so they act on $\sum_i c_i^{\,n+1}$ and
$\sum_i c_i^{\,n}$. The explicit potential contributes
$\nabla_h^2(\sum_i M_i N_i^{\,n}) = 0$ by \cref{eq:consistency}, pointwise and
unconditionally. Linear interpolation commutes with the sum, so the advective
fluxes sum to $\vu^{\,n}$ and their divergence vanishes. Hence
$\mathcal{L}(\sum_i c_i^{\,n+1}) = \mathcal{L}(\sum_i c_i^{\,n})$ with
$\mathcal{L}$ invertible, and \cref{eq:ic} gives $\sum_i c_i^{\,0} = 1$.
\hfill$\square$

The third concerns the pressure. Variable density would destroy the
constant-coefficient Poisson structure, so we use the split of \cite{dodd2014}
with $\rho_0 = \min_i\rho_i$ and a lagged pressure $\hat p$, iterated three
times per step. One detail interacts with \Cref{prop:simplex-discrete}: the
velocity correction must reuse the same $\hat p$ that entered the Poisson
right-hand side, since substituting the converged pressure into both terms
leaves $\nabla_h\!\cdot\!\vu^{\,n+1} = \mathcal{O}(10^{-5})$ and destroys the
exact simplex property.

None of the three introduces data-dependent control flow, so the whole step is
fused into one jitted JAX \cite{jax2018} kernel and \texttt{vmap}ped over
members: one compilation serves all $1{,}024$ runs, completed in $37$ minutes.
The scheme being explicit in advection, viscosity and capillarity, the binding
stability limit over the whole design box is viscous,
$\Delta t \le 0.2h^2/\nu_{\max} = 4.60\times10^{-4}$ at $\mathrm{Re} = 20$,
against $1.72\times10^{-3}$ for the capillary limit, so
$\Delta t = 2.5\times10^{-4}$ retains a $1.8\times$ margin.

\begin{table}[ht!]
\centering
\caption{Discretization, fixed physical quantities and cost. With
\Cref{tab:sobol} this specifies the benchmark completely.}
\label{tab:fixed}
\begin{tabular}{ll}
\toprule
Domain $\Omega$                  & $(0,1)\times(0,2.5)$ \\
Simulation grid                  & $128 \times 320$ \ ($h = 1/128$) \\
Output grid (2$\times$ block mean) & $64 \times 160$ \\
Bubble diameter $d$              & $0.25$ \\
Interface height $y_{\mathrm{int}}$ & $0.9$ \\
Gravity $g$, water density $\rho_w$ & $0.98$, \ $1$ \\
Velocity scale $U=\sqrt{gd}$     & $0.4950$ \\
Interface thickness              & $\varepsilon = 4h = 3.125\times10^{-2}$
                                   \ ($\mathrm{Cn}=\varepsilon/L_x$) \\
Mobility                         & $M_i = M_0/\Sigma_i$, \
                                   $M_0 = 10^{-3}\,U\varepsilon L_x
                                   = 1.55\times10^{-5}$ \\
CH stabilization                 & $S_i = S_0\Sigma_i$, \ $S_0 = 24/\varepsilon = 768$ \\
Triple-well penalty              & $\Lambda = 0$ \\
Kinematic viscosity ratio        & $\nu_a/\nu_w = 2$ \\
Pressure sub-iterations          & $3$ \\
Time step $\Delta t$             & $2.5\times10^{-4}$ (viscosity limited) \\
Steps to $t = 4$                 & $16{,}000$ \\
Snapshots                        & $33$, spacing $\Delta t_{\mathrm{out}} = 0.125$ \\
Ensemble wall clock              & $37$ min for $1{,}024$ runs \\
\bottomrule
\end{tabular}
\end{table}

\subsection{Initial condition and ensemble design}
\label{app:data:design}

Each member starts at rest with a circular air bubble of radius $d/2$ centred
at $(x_b,y_b)$ beneath a flat interface at $y_{\mathrm{int}}$. With
$\delta(\vx) = \|\vx-(x_b,y_b)\| - d/2$,
\begin{equation}
c_a^0 = \tfrac12\Big[1-\tanh\!\big(2\delta/\varepsilon\big)\Big],
\quad
c_w^0 = \big(1-c_a^0\big)\cdot
        \tfrac12\Big[1-\tanh\!\big(2(y-y_{\mathrm{int}})/\varepsilon\big)\Big],
\quad
c_o^0 = 1 - c_w^0 - c_a^0 ,
\label{eq:ic}
\end{equation}
the third field by affine closure exactly as in the surrogate
(\cref{eq:closure}), so \Cref{prop:simplex-discrete} applies from the first
step. Because \cref{eq:ic} is a closed-form function of $\vp$, the surrogate's
nonzero error at $t = 0$ (\Cref{sec:res:sample}) is pure trunk representation
error.

The design is a scrambled Sobol sequence \cite{sobol1967} of
$N = 1{,}024 = 2^{10}$ points over the nine parameters of \Cref{tab:sobol},
drawn in a single base-two call so the balance properties are preserved.
Because \cref{eq:tensions} guarantees $\Sigma_i > 0$ over the whole box, no
draw is rejected and the realized design is exactly the low-discrepancy
sequence. The domain height was set by a pre-flight run of the corner most
likely to reach the lid, which puts the plume tip at $y = 1.23$ at $t = 4$
against $L_y = 2.5$.

\begin{table}[ht!]
\centering
\caption{The nine Sobol design variables and their sampling ranges. Parameters
1--3 control capillarity and wetting, 4--7 the material contrasts, 8--9 the
initial bubble placement.}
\label{tab:sobol}
\begin{tabular}{cllll}
\toprule
\# & Symbol & Meaning & Range & Sets \\
\midrule
1 & $\mathrm{Eo}$    & E\"otv\"os number $\rho_w g d^2/\sigma_{aw}$ & $[5,\,40]$      & $\sigma_{aw}$ \\
2 & $\mathrm{Re}$    & Reynolds number $\rho_w U d/\eta_w$          & $[20,\,70]$     & $\eta_w$ \\
3 & $s$              & normalized oil spreading coeff.\ $\Sigma_o/\sigma_{aw}$ & $[0.02,\,0.50]$ & $\sigma_{ao}$ \\
4 & $r_{ow}$         & tension ratio $\sigma_{ow}/\sigma_{aw}$      & $[0.30,\,0.70]$ & $\sigma_{ow}$ \\
5 & $\rho_o/\rho_w$  & oil-to-water density ratio                   & $[0.70,\,0.95]$ & $\rho_o$ \\
6 & $\rho_a/\rho_w$  & air-to-water density ratio                   & $[0.01,\,0.10]$ & $\rho_a,\ \eta_a$ \\
7 & $\eta_o/\eta_w$  & oil-to-water viscosity ratio                 & $[0.5,\,3.0]$   & $\eta_o$ \\
8 & $y_b$            & initial bubble centre height                 & $[0.35,\,0.50]$ & IC \\
9 & $x_b$            & initial bubble centre lateral position       & $[0.45,\,0.55]$ & IC \\
\bottomrule
\end{tabular}
\end{table}

\subsection{Diagnostics, well-posedness and known imperfections}
\label{app:data:diag}

Four scalar diagnostics are stored at every snapshot,
\begin{align}
y_{\mathrm{bubble}} &= \frac{\int_\Omega c_a\,y\;d\vx}{\int_\Omega c_a\;d\vx},
&
A_{\mathrm{sharp}} &= \big|\{\vx: c_a(\vx)>\tfrac12\}\big| ,
\label{eq:diag1}\\[2pt]
V_{\mathrm{plume}} &= \int_{y>y_{\mathrm{int}}} c_w \; d\vx ,
&
y_{\mathrm{top}} &= \max\big\{y > y_{\mathrm{int}} : \textstyle\max_x c_w(x,y) > 0.3\big\} ,
\label{eq:diag2}
\end{align}
serving two purposes: they verify that the ensemble is well posed as an
operator-learning target, and they provide the accuracy measures of
\Cref{sec:res:physical}.

\begin{figure}[H]
\centering
\includegraphics[width=0.95\textwidth]{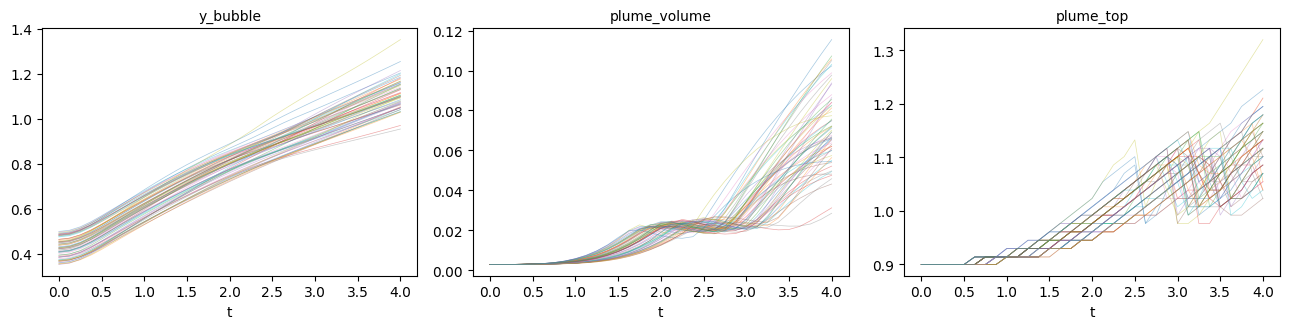}
\caption{Ensemble diagnostics across all members, from
\cref{eq:diag1,eq:diag2}. The response surfaces are smooth and
non-bifurcating: every member follows the same qualitative sequence --- rise,
deform, breakthrough, entrain, detach.}
\label{fig:diagnostics}
\end{figure}

\Cref{fig:diagnostics} shows what \cref{eq:operator} requires: a sweep across
$\mathrm{Eo}\in[5,40]$ gives a monotone response, with every case crossing the
interface and entraining a plume, so the design contains a single topological
class and $\Gop$ is single-valued. It is also a limitation, since the
surrogates are never tested near a regime boundary. Influence is deliberately
non-uniform --- a one-at-a-time study ranks $\mathrm{Eo}$ ($15.8\%$),
$\mathrm{Re}$ ($12.8\%$), $\rho_o$ ($12.3\%$), $\eta_o$ ($9.1\%$), $\rho_a$
($6.3\%$), $r_{ow}$ ($3.0\%$) and $s$ ($1.4\%$) --- since a surrogate's ability
to ignore a nuisance direction is worth testing. The placement parameters
register $18\%$ and $9.5\%$, but these are artifacts of a pointwise metric
under translation and are not physical sensitivities.

Each run is stored at $33$ snapshots over $t\in[0,4]$, coarsened by a
$2\times2$ block mean to $64\times160$, five channels in single precision, with
velocities interpolated from MAC faces to cell centres first. Block averaging
is linear, so the simplex constraint survives it exactly; discrete
incompressibility does not, and the stored velocities are divergence free only
to interpolation accuracy --- a caveat for anyone imposing
$\nabla\!\cdot\!\widehat{\vu} = 0$ on the coarse grid.

Two imperfections bear on \Cref{sec:results}. The simplex drift is small but
nonzero, at most $2.1\times10^{-4}$ over the ensemble, consistent with
single-precision accumulation over $16{,}000$ steps; it is a floating-point
residual, and it sets the bar against which the surrogate's own violation is
reported. And the bubble slowly dissolves: measured by $A_{\mathrm{sharp}}$,
the ensemble loses $15\%$ of its initial bubble area on average by $t = 4$.
This is the Ostwald-ripening-like dissolution of small features under
Cahn--Hilliard dynamics at finite $\varepsilon$ and mobility, a property of the
diffuse-interface model that constitutes the learning target rather than a
defect of the solver, and the price of $\varepsilon = 4h$, which is what allows
$1{,}024$ members in $37$ minutes.
\end{document}